\documentclass[journal,10pt,twoside]{IEEEtran}
\usepackage{amsmath,afterpage}
\usepackage{amsfonts}
\usepackage{amssymb}
\usepackage{amsthm}
\usepackage[usenames,dvipsnames]{xcolor}
\usepackage{graphicx, color, transparent}
\usepackage{tabularx}
\usepackage{makecell}
\usepackage{multirow}
\usepackage{pbox}
\usepackage{tikz}
\usetikzlibrary{plotmarks}
\usepackage{pgfplots}
\usetikzlibrary{calc}
\usetikzlibrary{shapes,arrows}
\usetikzlibrary{decorations.markings}
\pgfplotsset{compat=1.10}
\usepackage{amsthm}
\usepackage{cite}

\makeatletter
\newtheorem*{rep@theorem}{\rep@title}
\newcommand{\newreptheorem}[2]{%
\newenvironment{rep#1}[1]{%
 \def\rep@title{\Cref{##1}}%
 \begin{rep@theorem}}%
 {\end{rep@theorem}}}
\newcommand*{\textlabel}[2]{%
  \edef\@currentlabel{#1}
  #1\label{#2}
}
\makeatother

\makeatletter
\newcommand{\wast}{\bBigg@{2}}
\newcommand{\Wast}{\bBigg@{3}}
\newcommand{\vast}{\bBigg@{4}}
\newcommand{\Vast}{\bBigg@{5}}
\makeatother

\newtheorem{lemma}{Lemma}

\newtheorem{remark}{Remark}

\theoremstyle{definition}
\newtheorem{definition}{Definition}

\theoremstyle{plain}
\newtheorem{theorem}{Theorem}
\newreptheorem{theorem}{Theorem}
\newreptheorem{lemma}{Lemma}

\begin{document}

\title{Controllable Identifier Measurements for Private Authentication with Secret Keys}
\IEEEoverridecommandlockouts

\author{Onur G\"unl\"u,~\IEEEmembership{Student Member,~IEEE,}
		Kittipong Kittichokechai,~\IEEEmembership{Member,~IEEE,} Rafael F. Schaefer,~\IEEEmembership{Senior Member,~IEEE,} and Giuseppe Caire,~\IEEEmembership{Fellow,~IEEE}

\thanks{Manuscript received July 20, 2017; revised December 18, 2017; and accepted February 5, 2018. The work of O. G\"unl\"u was supported by the German Research Foundation (DFG) through the HoliPUF Project under Grant KR3517/6-1. The work of G. Caire was supported by an Alexander von Humboldt Professorship. Part of this paper was presented at the 2016 Asilomar Conference on Signals, Systems, and Computers \cite{bizimAsilomar}. The associate editor coordinating the review of this manuscript and approving it for publication was Dr. Tanya Ignatenko (\textit{Corresponding Author: Onur G\"unl\"u}).}
\thanks{O. G\"unl\"u is with the Chair of Communications Engineering, Technical University of Munich, 80333 Munich, Germany (e-mail: onur.gunlu@tum.de).
	
	 K. Kittichokechai was with the Communications and Information Theory Chair, Technische Universit\"at Berlin, 10623 Berlin, Germany. He is now with the Ericsson Research, 164 83 Stockholm, Sweden (e-mail: kittipong.kittichokechai@ericsson.com).
	 
	 R. F. Schaefer is with the Information Theory and Applications Chair, Technische Universit\"at Berlin, 10623 Berlin, Germany (email: rafael.schaefer@tu-berlin.de).
	 
     G. Caire is with the Communications and Information Theory Chair, Technische Universit\"at Berlin, 10623 Berlin, Germany (email: caire@tu-berlin.de).
	  
	 Digital Object Identifier 10.1109/TIFS.2018.2806937
	  }
}

\maketitle

\begin{abstract}
	The problem of secret-key based authentication under a privacy constraint on the source sequence is considered. The identifier measurements during authentication are assumed to be controllable via a cost-constrained ``action" sequence. Single-letter characterizations of the optimal trade-off among the secret-key rate, storage rate, privacy-leakage rate, and action cost are given for the four problems where noisy or noiseless measurements of the source are enrolled to generate or embed secret keys. The results are relevant for several user-authentication scenarios including physical and biometric authentications with multiple measurements. Our results include, as special cases, new results for secret-key generation and embedding with action-dependent side information without any privacy constraint on the enrolled source sequence. 
\end{abstract}


\IEEEpeerreviewmaketitle

\section{Introduction}\label{sec:introduction}
\IEEEPARstart{W}{e} study a problem of private authentication based on key generation or embedding, motivated by emerging technologies such as biometric authentication \cite{RaneSurvey} and key generation from physical unclonable functions (PUFs) \cite{pufintheory}. The system consists of an encoder and a decoder that observe different measurements of an identifier output and want to agree on a key, secret from an eavesdropper. 

Replacing biometric identifiers is generally impossible \cite{IgnaTrans}, and replacing physical identifiers is expensive or unappealing, for instance, if the new identifier outputs and the replaced ones are dependent. Therefore, for such applications, privacy of the identifier output is of significant importance because the biometric or physical source is closely related to the identity of a person or a device. There exists a fundamental trade-off between privacy and security performance of an authentication system. An information theoretic formulation provides a framework to capture such a trade-off \cite{IgnaTrans}, \cite{LaiTrans}. Moreover, the identifier measurements can be controlled or tuned with an additional cost. In this work, we study the optimal trade-offs among the secret-key rate, public storage rate, privacy-leakage rate, and expected action cost for discrete memoryless sources and measurement channels. Availability of post-processing methods in, e.g., \cite{bizimtemperature} to obtain memoryless channels and sources from biometric or physical identifiers allows us to not consider channels with memory and correlated sources, which are considered, e.g., in \cite{RennerWolf} and \cite{TyagiWatanabe}.

\subsection{Motivation}
The use of authentication for access control is an effective method to ensure information security. Unlike concealing the data to be transmitted \cite{WTC}, authentication of a user by using a secret requires correlated random variables in order to agree on a sequence \cite{AhlswedeCsiz,Maurer}. Most important physical identifiers used for device authentication are PUFs, e.g., random variations in ring oscillator (RO) outputs or in speckle patterns of optical tokens when irradiated by a laser. Similarly, body traits like irises and fingerprints are used as biometric randomness sources for authentication. There are code constructions in the biometric secrecy literature proposed for authentication, e.g., the fuzzy-vault scheme \cite{Fuzzyvault}, fuzzy-commitment scheme \cite{FuzzyCommitment}, and (code-offset) fuzzy extractors \cite{Dodis2008fuzzy}. It is shown in \cite{WZCodingproofarxiv} that the fuzzy-commitment scheme and fuzzy extractors are suboptimal for a simplified version of the private authentication problem we consider in this work. Accordingly, we are interested in understanding the fundamental limits of private authentication by studying optimal code constructions and their rate regions. 

Motivated by the use of biometric or physical identifiers that involve different forms of measurements, e.g., the use of multiple measurements or variations in the quality of the measurement process \cite{permuter}, \cite{Kittipong_a}, we consider a new private authentication model where the measurement process is represented by a cost-constrained action-dependent side information acquisition, where an action sequence determines the measurement channel. A high action cost can, for instance, represent the use of a high quality measurement device. 

There are two canonical models for private authentication: generated-secret model and chosen-secret model. We first consider the \emph{generated-secret model}, where the secret key is generated from the identifier outputs. The secret key reconstructed at the decoder is generally stored in a trusted database. It can therefore be practical to embed a uniformly-distributed and independently chosen secret key into the encoder rather than generating it from identifier outputs \cite{IgnaTrans}. The encoder binds the key to the identifier outputs in order to provide private authentication at the decoder. We also consider this practical model, called the \emph{chosen-secret model}, with cost-constrained actions and show that its general implication is an increased need of storage. Remark that the fuzzy-commitment scheme and fuzzy extractors are realizations of, respectively, the chosen- and generated-secret models.

Biometric and physical identifier outputs are noisy by nature. For instance, a cut in the palm corresponds to noise on the palmprint. Similar to multiple-antenna systems, multiple identifier measurements at the decoder can therefore significantly improve the rate regions as compared to a single measurement. Suppose we have multiple measurements also at the encoder, which assumes that the source is hidden or remote. A hidden or remote source represents that the encoder observes one or multiple noisy measurements of a source rather than the source output. It is shown in \cite{bizimproofarxiv} that if a visible source is mistakenly assumed for system design, there can be unnoticed secrecy leakages and the reliability at the decoder can decrease. Motivated by these results, we study also hidden identifiers with cost-constrained actions for the generated- and chosen-secret models.

\subsection{Summary of Contributions and Organization} 
 In \cite{bizimproofarxiv}, the enlargements of the rate regions due to increasing multiplicity of noisy measurements of a hidden source are illustrated. An attacker with access to a correlated identifier measurement tries to deceive the authentication in \cite{Kittipong_b}. We combine and extend the models in \cite{bizimproofarxiv} and \cite{Kittipong_b}, and consider a cost-constrained action sequence that controls the source measurements during authentication to reconstruct the secret key. In this work, the secret key can be either generated or embedded. Multiple identifier measurements both at the encoder and decoder are also possible by considering a hidden identifier. Similar to \cite{Kittipong_b}, correlated information at the eavesdropper is also considered here unlike in \cite{IgnaTrans}, \cite{LaiTrans}, and \cite{bizimproofarxiv}, which is a realistic assumption especially for biometric identifiers. The key-storage-leakage-cost region for secret-key generation from an identifier with a cost-constrained action at the decoder and a noiseless (visible) output at the encoder is given first in the conference version of this paper \cite{bizimAsilomar}. This rate region recovers several results in the literature including the key-leakage rate regions for a visible source in \cite{IgnaTrans} and \cite{LaiTrans}.

In this work, we further study the following extensions and the main contributions are as follows.
\begin{itemize}
	\item We extend the region for key generation to a chosen secret-key embedding scenario, where the source output is used to conceal the chosen secret key. 
	\item For a hidden source, we show that the key-storage-leakage-cost region is significantly different from the visible source model for both key generation and embedding scenarios. Comparisons among these regions illustrate that an incorrect system model could result in secrecy and reliability threats.
	\item As an example, we use realistic channel and source models to generate secret keys from PUFs and illustrate the key-leakage trade-off for a binary physical identifier with cost-constrained actions during authentication.     
\end{itemize}

This paper is organized as follows. In Section~\ref{sec:problem_setting}, we describe the source models and the generated- and chosen-secret models. We develop the key-storage-leakage-cost regions for the four problems, and compare them with each other and previous results in Section~\ref{sec:regions}. An achievable key-storage-leakage-cost region for a binary source with cost-constrained measurements during authentication is illustrated in Section~\ref{sec:examples}.


\subsection{Notation}
Upper case letters represent random variables and lower case letters their realizations. Superscripts denote a string of variables, e.g., $\displaystyle X^n\!=\!X_1\ldots X_i\ldots X_n$, and subscripts denote the position of a variable in a string. $X^{n\setminus i}$ represents the vector $(X_1,X_2,\ldots,X_{i-1},X_{i+1},\ldots,X_n)$. A random variable $\displaystyle X$ has probability distribution $\displaystyle P_X$. Calligraphic letters such as $\displaystyle \mathcal{X}$ denote sets and their sizes are written as $\displaystyle |\mathcal{X}|$. A set, e.g., $\mathcal{X}^n$, with superscript $n$ denotes an $n$-fold product-distribution set, and a set, e.g., $\mathcal{W}^{(n)}$, with superscript in parentheses $(n)$ denotes a set whose size grows with the superscript $n$. $\displaystyle \mathcal{T}_{\epsilon}^{n}(\cdot)$ denotes the set of length-$n$ letter-typical sequences with respect to the positive number $\displaystyle \epsilon$ \cite[Ch. 3]{masseylecturenotes}, \cite{orlitsky}. $X-Y-Z$ indicates that $(X,Y,Z)$ forms a Markov chain. $H_b(x)\!=\!-x\log x\!-\! (1\!-\!x)\log (1\!-\!x)$ is the binary entropy function and $\displaystyle H_b^{-1}(\cdot)$ denotes its inverse with range $[0,0.5]$. The $*$-operator is defined as $\displaystyle p\!*\!x\! =\! p(1\!-\!x)\!+\!(1\!-\!p)x$.

\section{Problem Formulations}\label{sec:problem_setting}
We define the four problems in the following.

\begin{figure}
	\centering
	\resizebox{0.75\linewidth}{!}{
		\begin{tikzpicture}
		\node (so) at (-1.5,-3.5) [draw,rounded corners = 5pt, minimum width=1.0cm,minimum height=0.8cm, align=left] {$P_X$};
		\node (a) at (0,-1.5) [draw,rounded corners = 6pt, minimum width=3.2cm,minimum height=0.6cm, align=left] {$
			(W, K) \overset{(a)}{=}f_1^{(n)}(X^n)$\\$W \overset{(b)}{=} f_2^{(n)}(X^n,K)$};
		\node (c) at (5,-3.5) [draw,rounded corners = 5pt, minimum width=1.3cm,minimum height=0.6cm, align=left] {$P_{YZ|XA}$};
		\node (b) at (5,-1.5) [draw,rounded corners = 6pt, minimum width=3.2cm,minimum height=0.9cm, align=left] {$\hat{K} = g^{(n)}\left(W,Y^n\right)$\\\smallskip$A^n=f_{a}^{(n)}(W)$};
		\node (g) at (5,-5) [draw,rounded corners = 5pt, minimum width=1cm,minimum height=0.6cm, align=left] {EVE};
		\draw[decoration={markings,mark=at position 1 with {\arrow[scale=1.5]{latex}}},
		postaction={decorate}, thick, shorten >=1.4pt] (a.east) -- (b.west) node [midway, above] {$W$};
		\node (a1) [below of = a, node distance = 2cm] {$X^n$};
		\draw[decoration={markings,mark=at position 1 with {\arrow[scale=1.5]{latex}}},
		postaction={decorate}, thick, shorten >=1.4pt] ($(c.north)+(0.3,0)$) -- ($(b.south)+(0.3,0)$) node [midway, right] {$Y^n$};
		\draw[decoration={markings,mark=at position 1 with {\arrow[scale=1.5]{latex}}},
		postaction={decorate}, thick, shorten >=1.4pt] (so.east) -- (a1.west);
		\draw[decoration={markings,mark=at position 1 with {\arrow[scale=1.5]{latex}}},
		postaction={decorate}, thick, shorten >=1.4pt] (a1.north) -- (a.south);
		\draw[decoration={markings,mark=at position 1 with {\arrow[scale=1.5]{latex}}},
		postaction={decorate}, thick, shorten >=1.4pt] (a1.east) -- (c.west);
		\draw[decoration={markings,mark=at position 1 with {\arrow[scale=1.5]{latex}}},
		postaction={decorate}, thick, shorten >=1.4pt] ($(b.south)-(0.3,0)$) -- ($(c.north)-(0.3,0)$)node [midway, left] {$A^n$};
		\draw[decoration={markings,mark=at position 1 with {\arrow[scale=1.5]{latex}}},
		postaction={decorate}, thick, shorten >=1.4pt] (c.south) -- (g.north) node [midway, right] {$Z^n$};
		\node (a2) [above of = a, node distance = 1.8cm] {$K$};
		\node (b2) [above of = b, node distance = 1.8cm] {$\hat{K}$};
		\draw[decoration={markings,mark=at position 1 with {\arrow[scale=1.5]{latex}}},
		postaction={decorate}, thick, shorten >=1.4pt] (b.north) -- (b2.south);
		\draw[densely dotted, decoration={markings,mark=at position 1 with {\arrow[scale=1.5]{latex}}},
		postaction={decorate}, thick, shorten >=1.4pt]  ($(a2.south)+(0.3,0)$)-- ($(a.north)+(0.3,0)$) node [midway, right] {$(b)$};
		\draw[dashed, decoration={markings,mark=at position 1 with {\arrow[scale=1.5]{latex}}},
		postaction={decorate}, thick, shorten >=1.4pt]  ($(a.north)-(0.3,0)$)-- ($(a2.south)-(0.3,0)$) node [midway, left] {$(a)$};
		\draw[decoration={markings,mark=at position 1 with {\arrow[scale=1.5]{latex}}},
		postaction={decorate}, thick, shorten >=1.4pt] ($(a.east)+(1,0)$) -- ($(a.east)+(1,-3.5)$) -- ($(a.east)+(1,-3.5)$) -- (g.west);
		\end{tikzpicture}
	}
	\caption{A visible source: $(a)$ represents the generated-secret model with the encoder $f_1^{(n)}(\cdot)$ and $(b)$ represents the chosen-secret model with the encoder $f_2^{(n)}(\cdot,\cdot)$. The decoder and EVE measurements can be performed after observing the action sequence.}\label{fig:visible}
\end{figure}
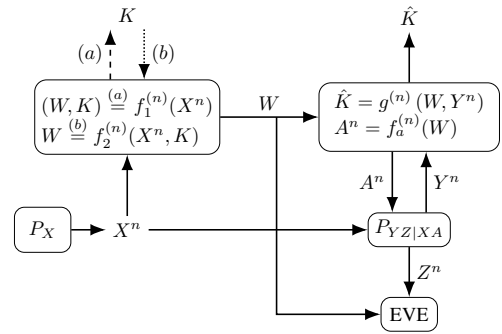

\subsection{Visible Source, Generated-secret Model}
Consider the system model in Fig. \ref{fig:visible}$(a)$. The source $\mathcal{X}$, measurements $\mathcal{Y},\mathcal{Z}$, and action $\mathcal{A}$ alphabets are finite sets. Let $X^n$ be a length-$n$ sequence which has independent and identically distributed (i.i.d.) components distributed according to some fixed distribution $P_{X}$. Authentication has two phases. First, a user enrolls the source sequence $X^n$ in the system to generate the helper data $W$ and the secret key $K$. A cost-constrained action sequence $A^n$ is chosen based on $W$ to control quality or reliability of the measurements during the authentication, during which $(Y^n,Z^n)$ are generated as outputs of a given memoryless channel $P_{YZ|XA}$ with inputs $X^n$ and $A^n$. The sequence $Y^n$ here represents a controllable measurement (side information) while $Z^n$ is another correlated side information.  Based on $W$ and measurement $Y^n$, the decoder reconstructs the secret key $\hat{K}$. Authentication is successful if $\hat{K}=K$. For generality, we consider an eavesdropper (EVE) who has access to the description $W$ and correlated side information $Z^n$.

\begin{definition}
   \normalfont	A $(|\mathcal{W}^{(n)}|,|\mathcal{K}^{(n)}|,n)$-code $\mathcal{C}_n$ for private authentication with a key generated from a visible source, controllable decoder measurements, and a noiseless encoder measurement consists of
	\begin{itemize}
		\item an encoder $f_1^{(n)}: \mathcal{X}^{n} \rightarrow \mathcal{W}^{(n)} \times \mathcal{K}^{(n)}$,
		\item an action encoder: $f_{a}^{(n)}: \mathcal{W}^{(n)} \rightarrow \mathcal{A}^{n}$,
		\item a decoder $g^{(n)}: \mathcal{W}^{(n)} \times \mathcal{Y}^{n} \rightarrow \mathcal{K}^{(n)}$. 
		\hfill $\lozenge$
	\end{itemize}
\end{definition}

\begin{definition}\label{def:achievabilitygeneratedsecret}
	\normalfont A key-storage-leakage-cost tuple $(R_k,R_w,\Delta,C)\in\mathbb{R}^{4}_{+}$ is said to be \emph{achievable} for a visible source with the generated-secret model if for any $\delta>0$ there is some $n\!\geq\!1$ and a $\big(|\mathcal{W}^{(n)}|, |\mathcal{K}^{(n)}|,n\big)$-code for which $\displaystyle R_k=\frac{\log|\mathcal{K}^{(n)}|}{n}$ such that
	\begin{align}
	&\Pr[\hat{K} \neq K] \leq \delta,&&\quad \quad(reliability)\label{eq:reliability_constraint}\\
	&\frac{1}{n} I(K;W,Z^n) \leq \delta&&\quad\quad(secrecy)\label{eq:secrecyleakage_constraint}\\
	&\frac{1}{n}H(K) \geq R_k- \delta&&\quad\quad(uniformity) \label{eq:uniformity_constraint}\\
	&\frac{1}{n}\log\big|\mathcal{W}^{(n)}\big| \leq R_w+\delta&&\quad\quad(storage)\label{eq:storage_constraint}\\
	&\frac{1}{n}I(X^n;W,Z^n) \leq \Delta+\delta&&\quad\quad(privacy)\label{eq:leakage_constraint}\\
	&\mathbb{E}[\Gamma^{(n)}(A^n)] \leq C+\delta&&\quad\quad(cost)\label{eq:cost_constraint}
	\end{align}
	where we have $(W,K)\! =\! f_1^{(n)}(X^n)$, $A^n \!=\! f_{a}^{(n)}(W)$, $\hat{K}\!=\!g^{(n)}(W,Y^n)$, and $\Gamma^{(n)}(\cdot)$ is a cost function with $\Gamma^{(n)}(A^n)\!=\!\frac{1}{n}\sum_{i=1}^n\Gamma(A_i)$. The \emph{key-storage-leakage-cost} region $\mathcal{R}_{gs}$ is the closure of the set of all achievable tuples.\hfill $\lozenge$
\end{definition}

\subsection{Visible Source, Chosen-secret Model}
Consider the problem of binding a secret key to a visible identifier, illustrated in Fig.~\ref{fig:visible}$(b)$. The decoder observes cost-constrained controllable measurements during authentication, whereas the encoder observes the noiseless source outputs. 

\begin{definition}
	\normalfont A $\big(|\mathcal{W}^{(n)}|, |\mathcal{K}^{(n)}|,n\big)$-code $\mathcal{C}_n$ for private authentication with an embedded key concealed by a visible source, controllable decoder measurements, and a noiseless encoder measurement consists of
	\begin{itemize}
		\item[$\bullet$] an encoder $f_2^{(n)}: \mathcal{X}^{n}\times \mathcal{K}^{(n)}\rightarrow \mathcal{W}^{(n)}$,
		\item[$\bullet$] an action encoder $f_{a}^{(n)}: \mathcal{W}^{(n)}\rightarrow\mathcal{A}^{n}$,
		\item[$\bullet$] a decoder $g^{(n)}: \mathcal{W}^{(n)}\times\mathcal{Y}^{n}\rightarrow \mathcal{K}^{(n)}$.\hfill$\lozenge$
	\end{itemize}
	
\end{definition}

\begin{definition} \label{def:achievabilitychosensecret}
	A key-storage-leakage-cost tuple $(R_k, R_w,\Delta,C)\in \mathbb{R}^4_{+}$ is said to be \emph{achievable} for a visible source with the chosen-secret model if for any $\delta\!>\!0$ there is some $n\!\geq\!1$ and a $\big(|\mathcal{W}^{(n)}|, |\mathcal{K}^{(n)}|,n\big)$-code for which $\displaystyle R_k\!=\!\frac{\log|\mathcal{K}^{(n)}|}{n}$ such that (\ref{eq:reliability_constraint})-(\ref{eq:cost_constraint}) are satisfied, where we have $W\!=\!f_2^{(n)}(X^n,K)$, $A^n\!=\!f_{a}^{(n)}(W)$, $\hat{K}\!=\!g^{(n)}(W,Y^n)$, and $\Gamma^{(n)}(A^n)\!=\!\frac{1}{n}\sum_{i=1}^n\Gamma(A_i)$. The \textit{key-storage-leakage-cost} region $\mathcal{R}_{cs}$ is the closure of all achievable tuples.\hfill$\lozenge$
\end{definition}

\subsection{Hidden Source, Generated-secret Model}
Consider the system model in Fig.~\ref{fig:hidden}$(a)$, where a key is generated from a hidden source. The decoder observes cost-constrained controllable source measurements $Y^n$ during authentication, whereas the encoder observes uncontrollable noisy measurements $\widetilde{X}^n$ of the hidden source outputs $X^n$ through a memoryless channel $P_{\widetilde{X}|X}$. The source alphabet $\mathcal{X}$, the measurement alphabets $\widetilde{\mathcal{X}},\mathcal{Y}, \mathcal{Z}$, and the action alphabet $\mathcal{A}$ are finite sets.

\begin{figure}
	\centering
	\resizebox{0.75\linewidth}{!}{
		\begin{tikzpicture}
		\node (so) at (-1.5,-3.5) [draw,rounded corners = 5pt, minimum width=1.0cm,minimum height=0.8cm, align=left] {$P_X$};
		\node (a) at (0,-0.5) [draw,rounded corners = 6pt, minimum width=3.2cm,minimum height=0.6cm, align=left] {$
			(W,K) \overset{(a)}{=} f_3^{(n)}(\widetilde{X}^n)$\\$W \overset{(b)}{=} f_4^{(n)}(\widetilde{X}^n,K)$};
		\node (c) at (5,-3.5) [draw,rounded corners = 5pt, minimum width=1.3cm,minimum height=0.6cm, align=left] {$P_{YZ|XA}$};
		\node (f) at (0,-2.25) [draw,rounded corners = 5pt, minimum width=1cm,minimum height=0.6cm, align=left] {$P_{\widetilde{X}|X}$};
		\node (b) at (5,-0.5) [draw,rounded corners = 6pt, minimum width=3.2cm,minimum height=0.9cm, align=left] {$\hat{K} = g^{(n)}\left(W,Y^n\right)$\\\smallskip$A^n=f_{a}^{(n)}(W)$};
		\node (g) at (5,-5) [draw,rounded corners = 5pt, minimum width=1cm,minimum height=0.6cm, align=left] {EVE};
		\draw[decoration={markings,mark=at position 1 with {\arrow[scale=1.5]{latex}}},
		postaction={decorate}, thick, shorten >=1.4pt] (a.east) -- (b.west) node [midway, above] {$W$};
		\node (a1) [below of = a, node distance = 3cm] {$X^n$};
		\draw[decoration={markings,mark=at position 1 with {\arrow[scale=1.5]{latex}}},
		postaction={decorate}, thick, shorten >=1.4pt] ($(c.north)+(0.3,0)$) -- ($(b.south)+(0.3,0)$) node [midway, right] {$Y^n$};
		\draw[decoration={markings,mark=at position 1 with {\arrow[scale=1.5]{latex}}},
		postaction={decorate}, thick, shorten >=1.4pt] (so.east) -- (a1.west);
		\draw[decoration={markings,mark=at position 1 with {\arrow[scale=1.5]{latex}}},
		postaction={decorate}, thick, shorten >=1.4pt] (a1.north) -- (f.south);
		\draw[decoration={markings,mark=at position 1 with {\arrow[scale=1.5]{latex}}},
		postaction={decorate}, thick, shorten >=1.4pt] (f.north) -- (a.south) node [midway, right] {$\widetilde{X}^n$};
		\draw[decoration={markings,mark=at position 1 with {\arrow[scale=1.5]{latex}}},
		postaction={decorate}, thick, shorten >=1.4pt] (a1.east) -- (c.west);
		\draw[decoration={markings,mark=at position 1 with {\arrow[scale=1.5]{latex}}},
		postaction={decorate}, thick, shorten >=1.4pt] ($(b.south)-(0.3,0)$) -- ($(c.north)-(0.3,0)$)node [midway, left] {$A^n$};
		\draw[decoration={markings,mark=at position 1 with {\arrow[scale=1.5]{latex}}},
		postaction={decorate}, thick, shorten >=1.4pt] (c.south) -- (g.north) node [midway, right] {$Z^n$};
		\node (a2) [above of = a, node distance = 1.8cm] {$K$};
		\node (b2) [above of = b, node distance = 1.8cm] {$\hat{K}$};
		\draw[decoration={markings,mark=at position 1 with {\arrow[scale=1.5]{latex}}},
		postaction={decorate}, thick, shorten >=1.4pt] (b.north) -- (b2.south);
		\draw[densely dotted, decoration={markings,mark=at position 1 with {\arrow[scale=1.5]{latex}}},
		postaction={decorate}, thick, shorten >=1.4pt]  ($(a2.south)+(0.3,0)$)-- ($(a.north)+(0.3,0)$) node [midway, right] {$(b)$};
		\draw[dashed, decoration={markings,mark=at position 1 with {\arrow[scale=1.5]{latex}}},
		postaction={decorate}, thick, shorten >=1.4pt]  ($(a.north)-(0.3,0)$)-- ($(a2.south)-(0.3,0)$) node [midway, left] {$(a)$};;
		\draw[decoration={markings,mark=at position 1 with {\arrow[scale=1.5]{latex}}},
		postaction={decorate}, thick, shorten >=1.4pt] ($(a.east)+(1,0)$) -- ($(a.east)+(1,-4.5)$) -- ($(a.east)+(1,-4.5)$) -- (g.west);
		\end{tikzpicture}
	}
	\caption{A hidden source: $(a)$ represents the generated-secret model with the encoder $f_3^{(n)}(\cdot)$ and $(b)$ represents the chosen-secret model with the encoder $f_4^{(n)}(\cdot,\cdot)$. The decoder and EVE measurements can be performed after observing the action sequence.}\label{fig:hidden}
\end{figure}
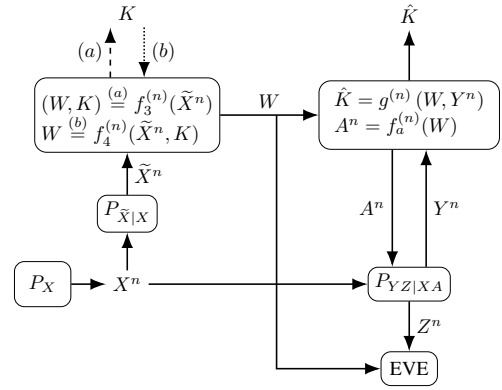

\begin{definition}
	\normalfont A $\big(|\mathcal{W}^{(n)}|, |\mathcal{K}^{(n)}|,n\big)$-code $\mathcal{C}_n$ for private authentication with a key generated from noisy measurements of a hidden source, controllable decoder measurements, and noisy encoder measurements consists of
	\begin{itemize}
		\item[$\bullet$] an encoder $f_3^{(n)}: \widetilde{\mathcal{X}}^{n}\rightarrow \mathcal{W}^{(n)}\times\mathcal{K}^{(n)}$,
		\item[$\bullet$] an action encoder $f_{a}^{(n)}: \mathcal{W}^{(n)}\rightarrow\mathcal{A}^{n}$,
		\item[$\bullet$] a decoder $g^{(n)}: \mathcal{W}^{(n)}\times\mathcal{Y}^{n}\rightarrow \mathcal{K}^{(n)}$.\hfill$\lozenge$
	\end{itemize}
\end{definition}

\begin{definition}
	A key-storage-leakage-cost tuple $(R_k, R_w,\Delta,C)\!\in\! \mathbb{R}^4_{+}$ is said to be \emph{achievable} for a hidden source with the generated-secret model if for any $\delta\!>\!0$ there is some $n\!\geq\!1$ and a $\big(|\mathcal{W}^{(n)}|, |\mathcal{K}^{(n)}|,n\big)$-code for which $\displaystyle R_k\!=\!\frac{\log|\mathcal{K}^{(n)}|}{n}$ such that (\ref{eq:reliability_constraint})-(\ref{eq:cost_constraint}) are satisfied, where we have $(W,K)\!=\!f_3^{(n)}(\widetilde{X}^n)$, $A^n\!=\!f_{a}^{(n)}(W)$, $\hat{K}\!=\!g^{(n)}(W,Y^n)$, and $\Gamma^{(n)}(A^n)\!=\!\frac{1}{n}\sum_{i=1}^n\Gamma(A_i)$. The \textit{key-storage-leakage-cost} region $\mathcal{R}_{hgs}$ is the closure of all achievable tuples.\hfill$\lozenge$
\end{definition}

\subsection{Hidden Source, Chosen-secret Model}

Consider the problem of binding a chosen secret key to a hidden biometric or physical identifier, as shown in Fig.~\ref{fig:hidden}$(b)$. The decoder observes cost-constrained controllable source measurements during authentication, whereas the encoder observes uncontrollable noisy source outputs. 

\begin{definition}
	\normalfont A $\big(|\mathcal{W}^{(n)}|, |\mathcal{K}^{(n)}|,n\big)$-code $\mathcal{C}_n$ for private authentication with an embedded secret key concealed by noisy measurements of a hidden source, controllable decoder measurements, and noisy encoder measurements consists of
	\begin{itemize}
		\item[$\bullet$] an encoder $f_4^{(n)}: \widetilde{\mathcal{X}}^{n}\times\mathcal{K}^{(n)}\rightarrow \mathcal{W}^{(n)}$,
		\item[$\bullet$] an action encoder $f_{a}^{(n)}: \mathcal{W}^{(n)}\rightarrow\mathcal{A}^{n}$,
		\item[$\bullet$] a decoder $g^{(n)}: \mathcal{W}^{(n)}\times\mathcal{Y}^{n}\rightarrow \mathcal{K}^{(n)}$.\hfill$\lozenge$
	\end{itemize}
\end{definition}

\begin{definition}
	A key-storage-leakage-cost tuple $(R_k, R_w,\Delta,C)\!\in\! \mathbb{R}^4_{+}$ is said to be \emph{achievable} for a hidden source with the chosen-secret model if for any $\delta\!>\!0$ there is some $n\!\geq\!1$ and a $\big(|\mathcal{W}^{(n)}|, |\mathcal{K}^{(n)}|,n\big)$-code for which $\displaystyle R_k\!=\!\frac{\log|\mathcal{K}^{(n)}|}{n}$ such that (\ref{eq:reliability_constraint})-(\ref{eq:cost_constraint}) are satisfied, where we have $W\!=\!f_4^{(n)}(\widetilde{X}^n, K)$, $A^n\!=\!f_{a}^{(n)}(W)$, $\hat{K}\!=\!g^{(n)}(W,Y^n)$, and $\Gamma^{(n)}(A^n)\!=\!\frac{1}{n}\sum_{i=1}^n\Gamma(A_i)$. The \textit{key-storage-leakage-cost} region $\mathcal{R}_{hcs}$ is the closure of all achievable tuples.\hfill$\lozenge$
\end{definition}

\begin{remark}
	\normalfont The encoder- and decoder-measurement channels in Fig.~\ref{fig:hidden} are modeled as two separate channels, i.e., $\widetilde{X}-(A,X)-(Y,Z)$ forms a Markov chain. This is the case if, e.g., there is a considerable amount of time between the encoder and decoder measurements of a palmprint so that the cuts on it during enrollment and authentication are independent.
\end{remark}

\section{Key-storage-leakage-cost Regions}\label{sec:regions}
We are interested in characterizing the optimal trade-off among the secret-key rate, storage rate, privacy-leakage rate, and expected action cost. We give the rate regions for all cases.

\begin{theorem}[\textit{Visible Source, Generated-secret}]\label{theo:gsvsregion}
	For given $P_{X}$ and $P_{YZ|XA}$, the key-storage-leakage-cost region $\mathcal{R}_{gs}$ is given as the set of all tuples $(R_k,R_w,\Delta,C) \in \mathbb{R}_{+}^4$ satisfying 
	\begin{align}
	&R_k \leq I(V;Y|A,U)-I(V;Z|A,U)\label{eq:keyrate}\\
	&R_w \geq I(X;A)+I(V;X|A,Y) \label{eq:storagerate}\\
	&\Delta\! \geq\! I(X;A,V,Y) \!+\! I(X;Z|A,U)\!-\!I(X;Y|A,U) \label{eq:leakagerate}
	\end{align}
	for some $P_{X}P_{A|X}P_{YZ|XA}P_{V|XA}P_{U|V}$ such that $\mathbb{E}[\Gamma(A)]\leq C$ with $|\mathcal{U}|\leq |\mathcal{X}||\mathcal{A}|+2$ and $|\mathcal{V}|\leq (|\mathcal{X}||\mathcal{A}|+2)(|\mathcal{X}||\mathcal{A}|+1)$.
\end{theorem}
\begin{IEEEproof}
	Achievability is based on a random coding scheme that consists of superposition of a rate-distortion code for communicating the action sequence and a layered coding with binning for secret-key generation. The converse is based on standard properties of entropy functions. The proof details are given in Appendices~\ref{appsub:achprooftheo1}-\ref{appsub:converseprooftheo1}.
\end{IEEEproof}

\begin{theorem}[\textit{Visible Source, Chosen-secret}]\label{theo:originalchosensecretregion}
	For given $P_X$ and $P_{YZ|XA}$, the key-storage-leakage-cost region $\mathcal{R}_{cs}$ is given as the set of all tuples $(R_k, R_w,\Delta,C)\in \mathbb{R}^4_{+}$ satisfying
	\begin{align}
	&R_k\leq I(V;Y|A,U)- I(V;Z|A,U)\label{eq:orcssecretrate}\\
	&R_w\geq I(X;A,V)-I(U;Y|A)-I(V;Z|A,U)\label{eq:orcsstoragerate}\\
	&\Delta\!\geq\!I(X;A,V,Y)\!+\! I(X;Z|A,U)\! -\! I(X;Y|A,U) \label{eq:orcsleakagerate}
	\end{align}
	for some $P_XP_{A|X}P_{YZ|XA}P_{V|XA}P_{U|V}$ such that $\mathbb{E}[\Gamma(A)]\!\leq\! C$ with $\displaystyle |\mathcal{U}|\!\leq\!|\mathcal{X}||\mathcal{A}|+2$ and $\displaystyle |\mathcal{V}|\!\leq\!(|\mathcal{X}||\mathcal{A}|+2)(|\mathcal{X}||\mathcal{A}|+1)$. 
\end{theorem}

\begin{IEEEproof}
We use the proof of achievability for Theorem~\ref{theo:gsvsregion} and add a one-time padding step. We apply the codebook generation and encoding steps of the generated-secret model to generate the key $K'$ and the helper data $W'$. The embedded chosen key $K$ is uniformly distributed and independent of other random variables. Compared to Theorem~\ref{theo:gsvsregion}, the secret-key and privacy-leakage rate bounds have the same expressions, and the storage rate bound is the sum of the secret-key and storage rate bounds of the generated-secret model. The proof details are given in Appendices~\ref{appsub:achprooftheo2}-\ref{appsub:converseprooftheo2}.
\end{IEEEproof}

\begin{remark}
	\normalfont The results in Theorems \ref{theo:gsvsregion} and ~\ref{theo:originalchosensecretregion} include, as special cases, results for one-round secret-key generation and embedding, respectively, that extend the results in \cite{AhlswedeCsiz}, where there is no privacy constraint on the source sequence, i.e.,  $\Delta = \Delta_{\max}= H(X)$ in Definitions~\ref{def:achievabilitygeneratedsecret} and ~\ref{def:achievabilitychosensecret}, with action-dependent side information. Moreover, Theorem \ref{theo:gsvsregion} can also be seen as an extension of the result in \cite{Kittipong_b} because we additionally capture cost-constrained action-dependent decoder measurements.
	
\end{remark}


\begin{theorem}[\textit{Hidden Source, Generated-secret}] \label{theo:hiddensourcegeneratedsecretregion}
	For given $P_X$, $P_{\widetilde{X}|X}$, and $P_{YZ|XA}$, the key-storage-leakage-cost region $\mathcal{R}_{hgs}$ is given as the set of all tuples $(R_k, R_w,\Delta,C)\in \mathbb{R}^4_{+}$ satisfying
	\begin{align}
	&R_k\leq I(V;Y|A,U)- I(V;Z|A,U)\label{eq:hiddensecretrate}\\
	&R_w\geq I(\widetilde{X};A) + I(V;\widetilde{X}|A,Y) \\
	&\Delta\!\geq\! I(X;A,V,Y)\! +\! I(X;Z|A,U)\!-\! I(X;Y|A,U) \label{eq:hiddenleakagerate}
	\end{align}
	for some $P_XP_{\widetilde{X}|X}P_{A|\widetilde{X}}P_{YZ|XA}P_{V|\widetilde{X}A}P_{U|V}$ such that $\mathbb{E}[\Gamma(A)]\!\leq\!C$ with $\displaystyle |\mathcal{U}|\!\leq\!|\mathcal{\widetilde{X}}||\mathcal{A}|+3$ and $\displaystyle |\mathcal{V}|\!\leq\!(|\mathcal{\widetilde{X}}||\mathcal{A}|+3)(|\mathcal{\widetilde{X}}||\mathcal{A}|+2)$. 
\end{theorem}

\begin{IEEEproof}
Achievability proof is similar to Theorem~\ref{theo:gsvsregion}. We mainly modify the privacy-leakage analysis since the source is now hidden. The proof is given in Appendices~\ref{appsub:achprooftheo3}-\ref{appsub:converseprooftheo3}.
\end{IEEEproof}

\begin{theorem}[\textit{Hidden Source, Chosen-secret}]\label{theo:hiddensourcechosensecretregion}
	For given $P_X$, $P_{\widetilde{X}|X}$, and $P_{YZ|XA}$, the key-storage-leakage-cost region $\mathcal{R}_{hcs}$ is given as the set of all tuples $(R_k, R_w,\Delta,C)\in \mathbb{R}^4_{+}$ satisfying
	\begin{align}
	&R_k\leq I(V;Y|A,U)- I(V;Z|A,U)\label{eq:hiddenchosensecretrate}\\
	&R_w\geq I(\widetilde{X};A,V)-I(U;Y|A)-I(V;Z|A,U)\label{eq:storagehiddensourcechosensecret}\\
	&\Delta\!\geq\! I(X;A,V,Y)\!+\!I(X;Z|A,U)\!-\!I(X;Y|A,U) \label{eq:hiddenchosenleakagerate}
	\end{align}
	for some $P_XP_{\widetilde{X}|X}P_{A|\widetilde{X}}P_{YZ|XA}P_{V|\widetilde{X}A}P_{U|V}$ such that $\mathbb{E}[\Gamma(A)]\!\leq\! C$ with $\displaystyle |\mathcal{U}|\!\leq\!|\mathcal{\widetilde{X}}||\mathcal{A}|+3$ and $\displaystyle |\mathcal{V}|\!\leq\!(|\mathcal{\widetilde{X}}||\mathcal{A}|+3)(|\mathcal{\widetilde{X}}||\mathcal{A}|+2)$. 
\end{theorem}

\begin{IEEEproof}
We use the proof of achievability for Theorem~\ref{theo:hiddensourcegeneratedsecretregion} and add a one-time padding step. The secret-key and privacy-leakage rate bounds have the same expressions, and the new storage rate bound is the sum of the secret-key and storage rate bounds of the generated-secret model for a hidden source. The proof details are given in Appendices~\ref{appsub:achprooftheo4}-\ref{appsub:converseprooftheo4}. 
\end{IEEEproof}

\begin{remark}
	\normalfont Theorems~\ref{theo:hiddensourcegeneratedsecretregion} and ~\ref{theo:hiddensourcechosensecretregion} can be seen as extensions of the results in \cite{bizimproofarxiv} with the addition of cost-constrained action-dependent measurements at the decoder and correlated side information at the eavesdropper.
\end{remark}

\subsection{Rate Region Comparisons and Discussions}
Consider the compression-leakage-key region given in \cite[Theorem 2]{Kittipong_b} for the generated-secret model and a visible source. We compare this region with the rate region $\mathcal{R}_{gs}$ to illustrate the effects of the cost-constrained action sequence. In particular, we observe that the action $A$ appears as a conditioning random variable in each mutual information term in \cite[Theorem 2]{Kittipong_b}, the new storage and privacy-leakage rate limits are increased by the rate-distortion coding amount of $\displaystyle I(X;A)$, and the probability distribution of $A$ is limited by an expected cost constraint. Therefore, the cost-constrained action sequence $A^n$ brings the possibility of enlarging the rate region, which recovers the rate region in \cite[Theorem 2]{Kittipong_b} by choosing a constant action with fixed cost. The action sequence $A^n$ has similar effects on other rate regions.

The rate region $\mathcal{R}_{gs}$ differs from the rate region $\mathcal{R}_{cs}$ only in the bound for the storage rate. The bound in (\ref{eq:orcsstoragerate}) can be written as $I(X;A) + I(X;V|A,Y)+ I(V;Y|A,U)- I(V;Z|A,U)$ (cf. (\ref{eq:storagerate})), revealing an additional rate that is $I(V;Y|A,U)- I(V;Z|A,U)$ (cf. (\ref{eq:orcssecretrate})) needed to convey the chosen secret to the decoder. Suppose $(R_k,R_w,\Delta,C)\in\mathcal{R}_{cs}$ for given $P_X$ and $P_{YZ|XA}$. Therefore, there exist $A$, $U$, and $V$ such that $U-V-(X,A)-(Y,Z)$ forms a Markov chain as in Theorem~\ref{theo:originalchosensecretregion}. It is straightforward to show that $(R_k,R_w-R_k,\Delta,C)\in\mathcal{R}_{gs}$ for the same $P_X$ and $P_{YZ|XA}$. Similar conclusions follow also for a hidden source.

The bounds for the secret-key and privacy-leakage rates of visible and hidden sources have the same expressions, i.e., 
for the generated-secret model in $\mathcal{R}_{gs}$ and $\mathcal{R}_{hgs}$, and for the chosen-secret model in $\mathcal{R}_{cs}$ and $\mathcal{R}_{hcs}$, respectively. However, the storage-rate limits of different source models are different. Moreover, the Markov chain constraints and the cardinality bounds on the auxiliary random variables are different for visible and hidden source models. The rate regions therefore differ significantly, which can result in unnoticed secrecy leakages and reliability reductions if the wrong source model is used for a system design (see \cite{bizimproofarxiv}).

\section{Example}\label{sec:examples}
 We want to illustrate an achievable rate region for cost-constrained action-dependent secret-key generation from a visible source. We first define the scenario where a PUF in an internet-of-things (IoT) device is used for key generation so that only a mobile device with access to the key can control the IoT device. We then show an achievable rate region for this scenario by proving specific convexity results. These convexity results significantly simplify the encoder design by decreasing the cardinality of the auxiliary random variable.  
 
 Suppose $X$ is binary and uniformly distributed, the channel $P_{A|X}$ is a binary symmetric channel (BSC) with crossover probability $\alpha$, and the channels $P_{Y|AX}(.|a,.)$ are BSCs with crossover probabilities $p_a$ for $a\!=\!0,1$. Suppose the eavesdropper has degraded side information and the channel $P_{Z|Y}$ is a BSC with crossover probability $p$. In practice, quantized fine variations of ring oscillator (RO) outputs follow these source and channel models. The effects of voltage and temperature variations can also be suppressed by a legitimate user by applying additional post-processing steps to the RO outputs \cite{bizimtemperature}. Classic crossover probabilities for the BSCs $P_{Y|AX}(\cdot|a,\cdot)$ under ideal environmental conditions are $p_a\!=\!0.03$ and $0.05$ for $a=0,1$, where, e.g., $a=0$ corresponds to the case that $X^n$ is sent through the $P_{Y|AX}(\cdot|0,\cdot)$ channel. 
 
 Suppose the attacker has access to a noisy version $Z^n$ of the RO outputs $X^n$ disturbed by environmental variations in addition to noise. A classic crossover probability for one of the BSCs $P_{Z|AX}(\cdot|a,\cdot)$ is $p^{\prime}\!=\!0.15$ \cite{bizimtemperature}. We thus choose $p_0\!=\!0.03$, $p_1\!=\!0.05$, $p\!=\!0.1277$ so that $p*p_0\!=\!0.15\!=\!p^{\prime}$ and $\displaystyle p*p_1\!=\!0.1649$. We also consider the cost of $\Gamma(0)\!=\!0.5$ units for $a\!=\!0$ and $\Gamma(1)\!=\!0.3$ units for $a\!=\!1$ since obtaining a more reliable channel requires more post-processing steps, which results in higher cost.

 Suppose the crossover probability $\alpha$ of the BSC $P_{A|X}$ is $0.2$. It is therefore more likely that the input $X\!=\!1$ is sent through a channel that is stochastically degraded with respect to the channel through which the input $X\!=\!0$ is sent because $p_1\!>\!p_0$. This is the case if, e.g., a one-bit quantizer is applied to RO outputs, where the bit $0$ is extracted if the output value is less than the mean over all ROs and the bit $1$ otherwise. RO outputs decrease with increasing temperature. Therefore, the error probability of the channel through which the input bit 0 is sent is smaller than the bit 1 is sent if the ambient temperature is greater than the temperature assumed for system design.
 
 We now illustrate an achievable rate region for the RO PUF problem defined above by proving convexity of a function used for entropy calculations. First, fix $V\!=\!(A,X)$ so that the rate region is
\begin{align}
&R_k\leq I(X;Y|A,U)-I(X;Z|A,U)\nonumber\\
&R_w\geq I(X;A)+H(X|A,Y)\nonumber\\
&\Delta\geq H(X)-(I(X;Y|A,U)-I(X;Z|A,U)) \label{eq:examplerateregion}
\end{align}
such that $\displaystyle U-(A,X)-(Y,Z)$ forms a Markov chain and $\displaystyle C\geq \mathbb{E}[\Gamma(A)]$. The optimization problem of achieving boundary points in (\ref{eq:examplerateregion}) is equivalent to
\begin{align}
\min_{P_{AX|U}}H(Z|A,U) \text{ for a fixed } H(Y|A,U)\!=\!\eta
\end{align}
for all $0\!\leq\!\eta\!\leq\!1$, which is a similar problem to Mrs. Gerber's lemma (MGL) \cite{WZ}. Denote the conditional probabilities $P_{AX|U}(ax|i)=\hat{x}_{i,ax}$ and the probabilities $P_U(i)\!=\!u_i$ for $i\!=\!1,2,\ldots,|\mathcal{U}|$. Due to $P_{AX}$, we obtain the constraints
\begin{align}
&\sum_{i=1}^{|\mathcal{U}|}u_i\hat{x}_{i,01} =\sum_{i=1}^{|\mathcal{U}|}u_i\hat{x}_{i,10} = \frac{\alpha}{2},\label{eq:PAXcondition1}\\
&\sum_{i=1}^{|\mathcal{U}|}u_i\hat{x}_{i,00} =\sum_{i=1}^{|\mathcal{U}|}u_i\hat{x}_{i,11} = \frac{1\!-\!\alpha}{2}\label{eq:PAXcondition2}.
\end{align}
To fix $H(Y|A,U)$, it therefore suffices to consider 
\begin{align}
\hat{x}_{i,01} =\frac{1}{2}\!-\!\hat{x}_{i,00}, \qquad \hat{x}_{i,10} =\frac{1}{2}\!-\!\hat{x}_{i,11}\label{eq:firstPAXsimpl12}.
\end{align}
Define the functions
\begin{align}
f(\hat{x}_{i,00},\hat{x}_{i,11})&\!=\!\Bigg[H_b\Bigg(p_0*\!\frac{2\hat{x}_{i,00}}{1-2(\hat{x}_{i,11}-\hat{x}_{i,00})}\Bigg)\nonumber\\
&\!+\!H_b\Bigg(\!p_1*\!\frac{2\hat{x}_{i,11}}{1\!-\!2(\hat{x}_{i,00}-\hat{x}_{i,11})}\!\Bigg)\!\Bigg],\label{eq:fdef}
\end{align}
\begin{align}
g(\hat{x}_{i,00},\hat{x}_{i,11})&\!=\!\Bigg[H_b\Bigg(p*p_0*\!\frac{2\hat{x}_{i,00}}{1-2(\hat{x}_{i,11}-\hat{x}_{i,00})}\Bigg)\nonumber\\
&\!+\!H_b\!\Bigg(\!p\!*p_1\!*\!\frac{2\hat{x}_{i,11}}{1\!-\!2(\hat{x}_{i,00}\!-\!\hat{x}_{i,11})}\!\Bigg)\!\Bigg]\label{eq:gdef}.
\end{align}

Using (\ref{eq:firstPAXsimpl12}), (\ref{eq:fdef}), and (\ref{eq:gdef}), we obtain
\begin{align}
H(Y|A,U) = \sum_{i=1}^{|\mathcal{U}|}&u_i\,\frac{1}{2}\, f(\hat{x}_{i,00},\hat{x}_{i,11}),\label{eq:hYgivenAU}
\end{align}
\begin{align}
H(Z|A,U) = &\sum_{i=1}^{|\mathcal{U}|}u_i\,\frac{1}{2}\, g(\hat{x}_{i,00},\hat{x}_{i,11})\label{eq:hZgivenAU}.
\end{align}
Define an inverse function $\displaystyle f^{-1}(\nu)\!=\!(\bar{x},\bar{x})$ for all $\nu\!\in\![H_b(p_0)\!+\!H_b(p_1),\, 2]$ and $\bar{x}\!\in\![0,\, 0.5]$. It suffices to replace $f(\hat{x}_{i,00},\hat{x}_{i,11})$ and $g(\hat{x}_{i,00},\hat{x}_{i,11})$, respectively, with
\begin{align}
&\bar{f}(\bar{x})=f\Big(\frac{\bar{x}}{2},\frac{\bar{x}}{2}\Big)\label{eq:functionfbardef}\\
&\bar{g}(\bar{x})=g\Big(\frac{\bar{x}}{2},\frac{\bar{x}}{2}\Big) \label{eq:functiongbardef}
\end{align}
to fix (\ref{eq:hYgivenAU}) and (\ref{eq:hZgivenAU}) separately. 

\begin{lemma}\label{lem:uniqueness}
	There is a unique $\bar{x}$ in the interval $[0,\,0.5]$ for which $H(Y|A,U)\!=\!\frac{1}{2}\bar{f}(\bar{x})$.
\end{lemma}
\begin{proof}
	The function $\bar{f}(\bar{x})$ is strictly increasing from $H_b(p_0)\!+\!H_b(p_1)$ to $2$ in the interval $[0,0.5)$ and we have $H_b(p_0)\!+\!H_b(p_1)\!\leq\!2H(Y|A,U)\!\leq\!2H(Y)\!\leq\!2$.
\end{proof}

\begin{lemma}\label{lem:convexity}
	Define $\tilde{p}'\!=\!\min\{p',1\!-\!p'\}$ for some $0\!\leq\! p'\!\leq\! 1$. If $\tilde{p}*\tilde{p}_0\geq\tilde{p}_1$ and $\tilde{p}*\tilde{p}_1\geq\tilde{p}_0$, the function $\bar{g}(f^{-1}(\nu))$ is convex in $\nu$ for $\nu\!\in\![H_b(p_0)\!+\!H_b(p_1),\,2]$.
\end{lemma}
\begin{proof}
	The functions $\bar{f}(\bar{x})$ and $\bar{g}(\bar{x})$ are symmetric with respect to $p_0\!=\!\frac{1}{2}$, $p_1\!=\!\frac{1}{2}$, and $p\!=\!\frac{1}{2}$. It thus suffices to prove the convexity for $\displaystyle 0\!\leq\!\tilde{p}_0,\tilde{p}_1,\tilde{p}\!\leq\!0.5$. Define $ \bar{f}^{\prime}(\bar{x})\! =\! \frac{d}{d\bar{x}}\bar{f}(\bar{x})$. $\bar{g}(f^{-1}(\nu))$ is convex in $\nu$ if \cite{bizimCNS}
	\begin{align}
	\frac{\partial^{2}}{\partial {\nu^{2}}} \left(\bar{g}(f^{-1}(\nu))\right)\!=\!\frac{1}{\bar{f}^{\prime}(\bar{x})}\frac{\partial}{\partial {\bar{x}}} \left(\frac{\bar{g}^{\prime}(\bar{x})}{\bar{f}^{\prime}(\bar{x})}\right)\!\geq\!0 \label{eq:seceq} 
	\end{align}
	for all $\bar{x}\!\in[0, 0.5]$. Note that $H_b(\cdot)$ is an increasing function for $\bar{x}\!\in[0, 0.5]$, so $\bar{f}^{\prime}(\bar{x})\!\geq\!0$ for all $\bar{x}\!\in[0, 0.5]$. It thus suffices to show that $\frac{\partial}{\partial {\bar{x}}} \left(\frac{\bar{g}^{\prime}(\bar{x})}{\bar{f}^{\prime}(\bar{x})}\right)\!\geq\!0$, i.e., 
	\begin{align}
	\bar{g}^{\prime\prime}(\bar{x})\bar{f}^{\prime}(\bar{x})\!-\!\bar{f}^{\prime\prime}(\bar{x})\bar{g}^{\prime}(\bar{x})\!\geq\!0 \label{eq:convexitynextstep}.
	\end{align}
	
	The functions $\bar{f}(\bar{x})$ and $\bar{g}(\bar{x})$ consist of two parts as $H_b(\tilde{p}_a*\bar{x})$ and $H_b(\tilde{p}*\tilde{p}_a*\bar{x})$, respectively, for $a\!=\!0,1$. It is shown in \cite{WZ} that $\displaystyle H_b(\tilde{p}*H_b^{-1}(\nu))$ is convex in $0\!\leq\!\nu\!\leq\!1$ for any $\tilde{p}\!\in\![0,\,0.5]$, so the terms in (\ref{eq:convexitynextstep}) that consist of the multiplications of the parts with the same $\tilde{p}_a$ provide positive contributions. It thus suffices to find a set of $\tilde{p}_0$ and $\tilde{p}_1$ values that satisfies
	\begin{align}
	&\frac{1\!-\!2(\tilde{p}\!*\!\tilde{p}_a)}{(\tilde{p}\!*\!\tilde{p}_a\!*\!\bar{x})(1\!-\!\tilde{p}\!*\!\tilde{p}_a\!*\!\bar{x})\log\Big(\frac{1-\tilde{p}*\tilde{p}_a*\bar{x}}{\tilde{p}*\tilde{p}_a*\bar{x}}\Big)}\nonumber\\
	&\leq\!\frac{1\!-\!2\tilde{p}_{b}}{(\tilde{p}_{b}\!*\!\bar{x})(1\!-\!\tilde{p}_{b}\!*\!\bar{x})\log\Big(\frac{1-\tilde{p}_{b}*\bar{x}}{\tilde{p}_{b}*\bar{x}}\Big)}\label{eq:convexitylaststepineq}
	\end{align}
	where $b\!=\!1\!-\!a$ for $a\!=\!0,1$. Define the function 
	\begin{align}
	l(\hat{p}) = \frac{1-2\hat{p}}{(\hat{p}\!*\!\bar{x})(1\!-\!\hat{p}\!*\!\bar{x})\log\Big(\frac{1-\hat{p}*\bar{x}}{\hat{p}*\bar{x}}\Big)}
	\end{align}
	for $0\!\leq\!\hat{p},\bar{x}\!\leq\!0.5$. It is straightforward to prove that $l(\hat{p})$ is a decreasing function by showing that $l(\hat{p})$ is convex and $ l^{\prime}(0.5)\!=\!0$. The inequality in (\ref{eq:convexitylaststepineq}) is thus satisfied if $\tilde{p}\!*\!\tilde{p}_a\!\geq \tilde{p}_b$ for $a\!=\!0,1$. This proves the convexity.
\end{proof}

We use the convexity property for channels satisfying the assumptions in Lemma~\ref{lem:convexity} to give an achievable lower bound for $H(Z|A,U)$ when $H(Y|A,U)$ is fixed.

\begin{lemma}\label{lem:lowerbound}
	Suppose $\bar{g}(f^{-1}(\nu))$ is convex in $\nu$. With the assumptions given above, we have 
	\begin{align}
	H(Z|A,U)\!\geq\!\frac{1}{2}\,\bar{g}(f^{-1}(2H(Y|A,U)))\label{eq:jensenresult}.
	\end{align}
\end{lemma}
\begin{proof}
	Using Jensen's inequality, we have
	\begin{align*}
	&H(Z|A,U)\!=\!\sum_{i=1}^{|\mathcal{U}|} u_i\,\frac{1}{2}\,\bar{g}(f^{-1}(\bar{f}(\bar{x}_i)))\nonumber\\
	&\!\geq\!\frac{1}{2}\,\bar{g}\Big(f^{-1}\big(\sum_{i=1}^{|\mathcal{U}|} u_i\bar{f}(\bar{x}_i)\big)\Big)\!=\!\frac{1}{2}\,\bar{g}\big(f^{-1}(2H(Y|A,U))\big).\qedhere
	\end{align*}
\end{proof}

\begin{lemma}\label{lem:result}
	Consider the problem setup defined above and the region in (\ref{eq:examplerateregion}). The BSCs $P_{AX|U}(a,\cdot|\cdot)$ with the same crossover probability $\bar{x}\!\in\![0,\,0.5]$ when $P_{AX|U}(a,0|\cdot)\!+\!P_{AX|U}(a,1|\cdot)\!=\!\frac{1}{2}$ achieve the region that satisfies equality in (\ref{eq:jensenresult}) if (\ref{eq:firstPAXsimpl12}), $\tilde{p}*\tilde{p}_0\!\geq\!\tilde{p}_1$, and $\tilde{p}*\tilde{p}_1\!\geq\!\tilde{p}_0$ are satisfied.
\end{lemma}
\begin{proof}
	Consider the boundary points in (\ref{eq:examplerateregion}) that depend on $U$. Using Lemma~\ref{lem:lowerbound}, we obtain
	\begin{align}
	&R_k\!\leq H(Y|A,U)\!-\!H(Y|A,X)\nonumber\\
	&\qquad\!-\!\frac{1}{2}\,\bar{g}\big(f^{-1}(2H(Y|A,U))\big)\!+\!H(Z|A,X)\label{theo:boundRk},\\
	&\Delta\!\geq H(X)\!-\!H(Y|A,U)\!+\!H(Y|A,X)\nonumber\\
	&\qquad\!+\!\frac{1}{2}\,\bar{g}\big(f^{-1}(2H(Y|A,U))\big)\!-\!H(Z|A,X)\label{theo:bounddelta}\
	\end{align}
	where we use Lemma~\ref{lem:convexity} for the convexity requirement and Lemma~\ref{lem:uniqueness} to show that the inverse function $f^{-1}(\cdot)$ is a bijective mapping. Equalities in (\ref{theo:boundRk}) and (\ref{theo:bounddelta}) are achieved by BSCs $P_{AX|U}(a,\cdot|\cdot)$ with crossover probability $0\!\leq\!\bar{x}\!\leq\!0.5$, defined in Lemma~\ref{lem:uniqueness}, when $P_{AX|U}(a,0|\cdot)\!+\!P_{AX|U}(a,1|\cdot)\!=\!\frac{1}{2}$. 
\end{proof}

\begin{remark}
	\normalfont One can show that the lower bound in (\ref{eq:jensenresult}) can be improved for $H(Z|A,U)$ given in (\ref{eq:hZgivenAU}) that is a function of a general $g(\hat{x}_{i,00},\hat{x}_{i,11})$, although this lower bound is tight for the function $\bar{g}(\bar{x})$. 
\end{remark}

For the RO PUF problem with the source and channel parameters given above, we obtain $R_w\!\geq\!0.4731$ bits/source-bit and $C\!\geq\!0.4$ units since $P_{AXYZ}$ is fixed.
The boundary points for $R_k$ and $\Delta$ sum up to $H(X)\!=\!1$ bits, which determines the trade-off between the secret-key and privacy-leakage rates for this example. The maximum $R_k$ achievable by using Lemma~\ref{lem:result} is $R_k^*\!=\!0.3876$ bits/source-bit, achieved with $\Delta\!\geq\!0.6124$ bits/source-bit.

\section{Conclusion}\label{sec:conclusion}
We derived the key-storage-leakage-cost regions for a visible source with the generated- or chosen-secret model when a cost-constrained action sequence controls the source measurements during authentication. Correlated side information at the eavesdropper is also considered as a realistic assumption especially for biometric identifiers. The achievability proof of the generated-secret model involves layered random binning. We bound the secret key generated by the generated-secret model to a chosen secret key for the proof of the chosen-secret model. We illustrated achievable key-storage-leakage-cost regions with an example, where used channel and source parameters were motivated by realistic authentication scenarios that use secret keys generated from RO PUFs.  

Multiple source measurements during enrollment are studied by considering a hidden source with noisy measurements at the encoder. We also derived the key-storage-leakage-cost regions for such a hidden source. The achievability proofs of the hidden source models also involve the same layered random binning as of the visible source models, but this time the noiseless identifier outputs are replaced with the noisy outputs at the encoder and the privacy-leakage rate is measured with respect to the hidden source. 

Comparisons showed that the rate regions for the two source models differ significantly due to different rate limits for the storage rate, and different Markov chain constraints and cardinality bounds on the auxiliary random variables. In future work, we will consider adaptive decoder measurements with causal actions that depend on the helper data and previous decoder measurements, which might improve the rate regions.

\section*{Acknowledgment}
The authors thank the Associate Editor and anonymous reviewers for their valuable suggestions that helped to improve the paper. Specifically, we thank an anonymous reviewer who suggested the future work problem above.

\appendix

\section*{Proofs of Theorems~\ref{theo:gsvsregion}-~\ref{theo:hiddensourcechosensecretregion}}\label{sec:proofs}
Based on the condition that all sequences are jointly typical with high probability, we bound some conditional entropy terms of interest with single letter expressions using the following two lemmas (see \cite{Kittipong_a} for proofs). 

\begin{lemma}\label{lemma:3}
	Let $(X^{n},A^{n})$ be jointly typical with high probability and $Z^n \ \text{i.i.d.} \sim P_{Z|XA}$, we have $H(Z^{n}|X^{n},A^{n}) \geq n(H(Z|X,A) - \delta_{\epsilon})$, where $\delta_{\epsilon} \rightarrow 0$ as $\epsilon \rightarrow 0$ and $\epsilon \rightarrow 0$ as $n \rightarrow \infty$.
\end{lemma}

\begin{lemma}\label{lemma:4}
	Let $(A^{n},U^{n},Z^n)$ be jointly typical with high probability and $\mathcal{C}_n$ represent a random codebook. Then, $H(Z^n|A^n,U^n,\mathcal{C}_n)\! \leq\! n(H(Z|A,U) \!+\! \delta_{\epsilon})$, where $\delta_{\epsilon} \rightarrow 0$ as $\epsilon \rightarrow 0$ and $\epsilon \rightarrow 0$ as $n \rightarrow \infty$.
\end{lemma}
\section*{Proof of Theorem~\ref{theo:gsvsregion}}\label{appsec:proofstheo1}
\subsection{Proof of Achievability}\label{appsub:achprooftheo1}
The proof follows from standard random coding arguments where we show the existence of a code that satisfies the key, storage, privacy-leakage rates, and expected cost constraints. 

\emph{Codebook generation}: Fix $P_{A|X}P_{V|XA}P_{U|V}$ such that $\mathbb{E}[\Gamma(A)]\leq C/(1+\epsilon)$.

\begin{itemize}
	\item Randomly and independently generate $2^{n(I(X;A)+\delta_{\epsilon})}$ codewords $a^{n}(w_a)$  according to  $\prod_{i=1}^{n}P_A(a_i(w_a))$ for $w_a \in [1:2^{n(I(X;A)+\delta_{\epsilon})}]$.
	\item For each $w_a$, randomly and conditionally independently generate $2^{n(I(U;X|A)+\delta_{\epsilon})}$ codewords $u^{n}(w_a,m)$ each according to $\prod_{i=1}^{n}P_{U|A}(u_i|a_i(w_a))$ for $m \in [1:2^{n(I(U;X|A)+\delta_{\epsilon})}]$, and distribute them uniformly at random into $2^{n(I(U;X|A)-I(U;Y|A)+2\delta_{\epsilon})}$ bins $b_U(w_u)$ for $w_u \in [1:2^{n(I(U;X|A)-I(U;Y|A)+2\delta_{\epsilon})}]$. Without loss of generality, we can identify the index $m = (w_u,m')$ for some $m' \in [1:2^{n(I(U;Y|A)-\delta_{\epsilon})}]$.
	\item For each $(w_a,m)$ pair, randomly and conditionally independently generate $2^{n(I(V;X|A,U)+\delta_{\epsilon})}$ codewords $v^{n}(w_a,m,l)$ each according to $\prod_{i=1}^{n}P_{V|UA}(v_i|u_i(w_a,m),a_i(w_a))$ for $ l \in [1:2^{n(I(V;X|A,U)+\delta_{\epsilon})}]$, and distribute them uniformly at random into $2^{n(I(V;X|A,U)-I(V;Y|A,U)+3\delta_{\epsilon})}$ bins $b_V(m,w_v)$ for $w_v \in [1:2^{n(I(V;X|A,U)-I(V;Y|A,U)+3\delta_{\epsilon})}]$. Furthermore, for each bin, we divide codewords $v^n$ into $2^{n(I(V;Y|A,U)-I(V;Z|A,U)-\delta_{\epsilon})}$ equal-sized subbins, each denoted by a subbin index $w_k$. Without loss of generality, we can identify the index $l = (w_v,w_k,l')$ for some $l' \in [1:2^{n(I(V;Z|A,U)-\delta_{\epsilon})}]$.
\end{itemize}
The codebook is revealed to all parties.

\emph{Encoding}:
\begin{itemize}
	\item For a given source sequence $x^{n}$, the encoder looks for a $a^{n}(w_a)$ which is jointly typical with $x^n$. Since there are more than $2^{nI(X;A)}$ codewords $a^n$, by the covering lemma \cite{Elgamalbook}, there exists such an $a^{n}$ with high probability. If there are more than one, we choose one uniformly at random and send the corresponding index $w_a$ to the decoder.
	\item The encoder then looks for a $u^{n}(w_a,m)$ that is jointly typical with $(x^n,a^{n})$. Since there are more than $2^{nI(U;X|A)}$ codewords $u^n$, by the covering lemma, there exists such a $u^{n}$ with high probability. If there are more than one, we choose one uniformly at random and send the corresponding bin index $w_u$ to the decoder.
	\item Again, the encoder looks for a $v^{n}(w_a,m,l)$ which is jointly typical with $(x^n,a^{n},u^n)$. Since there are more than $2^{nI(V;X|A,U)}$ codewords $v^n$, by the covering lemma, there exists such a $v^{n}$ with high probability. If there are more than one, we choose one uniformly at random and send the corresponding bin index $w_v$ to the decoder. The secret key $k$ is chosen to be the subbin index $w_k$ of the chosen codeword $v^n$.
\end{itemize}
This gives the total storage rate of $I(X;A)+I(U;X|A)-I(U;Y|A)+I(V;X|A,U)-I(V;Y|A,U)+6\delta_{\epsilon}= I(X;A)+I(V;X|A,Y)+6\delta_{\epsilon}$. Once the action sequence is chosen, action-dependent side information $(y^n,z^n)$ is generated as the output of the memoryless channel $P_{Y,Z|X,A}$.

\emph{Decoding}:
\begin{itemize}
	\item Upon receiving the indices $(w_a, w_u,w_v)$ 
	and side information $y^n$, the decoder looks for the unique $u^{n}$ which is jointly typical with $(y^{n}, a^n)$. Since there are less than $2^{nI(U;Y|A)}$ sequences in the bin $b_U(w_u)$, by the packing lemma \cite{Elgamalbook}, it will find the unique and correct $u^{n}$ with high probability.
	\item Then, the decoder looks for the unique $v^{n}$ which is jointly typical with $(y^{n}, a^n, u^n)$. Since there are less than $2^{nI(V;Y|A,U)}$ sequences in the bin $b_V(m,w_v)$, by the packing lemma, it will find the unique and correct $v^{n}$ with high probability. The decoder puts out $\hat{k}$ as the subbin index $\hat{w}_k$ of the decoded codeword $v^n$ which will be the correct one with high probability.
\end{itemize}

\emph{Action Cost}: Since each action sequence $a^n$ is in the typical set with high probability, by the typical average lemma \cite{Elgamalbook}, the expected cost constraint is satisfied.

\emph{Privacy-leakage Rate}: The information leakage  averaged over the random codebook $\mathcal{C}_n$ can be bounded as
\begin{align}
	& I(X^{n};W_a,W_u,W_v,Z^{n}|\mathcal{C}_n)\nonumber\\
	&\leq I(X^{n};W_a,M,W_v,Z^{n}|\mathcal{C}_n)\nonumber\\
	&= H(X^n|\mathcal{C}_n)-H(X^{n},W_a,M,W_v,Z^{n}|\mathcal{C}_n)\nonumber\\&\qquad +H(W_a,M,W_v|\mathcal{C}_n)+H(Z^n|W_a,M,W_v,\mathcal{C}_n)\nonumber\\
	&= -H(Z^{n}|X^{n},\mathcal{C}_n)-H(W_a,M,W_v|X^{n},Z^{n},\mathcal{C}_n)\nonumber\\&\qquad +H(W_a,M,W_v|\mathcal{C}_n)+H(Z^n|W_a,M,W_v,\mathcal{C}_n)\nonumber\\
	&\overset{(a)}{\leq}- H(Z^{n}|X^{n},A^{n}) +H(W_a,M,W_v|\mathcal{C}_n)\nonumber\\&\qquad+H(Z^n|W_a,M,W_v,\mathcal{C}_n)\nonumber\\
	&\overset{(b)}{\leq} -H(Z^{n}|X^{n},A^{n})+H(W_a|\mathcal{C}_n)+H(M|\mathcal{C}_n)\nonumber\\&\qquad+H(W_v|\mathcal{C}_n)+H(Z^n|A^n,U^n,\mathcal{C}_n)\nonumber\\
	& \overset{(c)}{\leq} n[- H(Z|X,A)+I(X;A)+I(U;X|A)+5\delta_{\epsilon}\nonumber\\&\quad +(I(V;X|A,U)-I(V;Y|A,U))+H(Z|A,U)]\nonumber\\
	& \!\overset{(d)}{=}\! n[I(X;A,V,Y)\!-\!I(X;Y|A,U)\!+\!I(X;Z|A,U)\!+\!\delta_{\epsilon}']\nonumber \\
	&\leq n[\Delta +\delta_{\epsilon}']\label{eq:achprivacyleakagegeneratedsecret}
\end{align}
if $\Delta \geq I(X;A,V,Y)-(I(X;Y|A,U)-I(X;Z|A,U))$,
where $(a)$ follows from the facts that conditioning reduces entropy, 
and that $Z^n-(X^n,A^n)-\mathcal{C}_n$ forms a Markov chain, $(b)$ follows because given the codebook, $(A^n,U^n)$ are functions of $(W_a,M)$, $(c)$ follows from the codebook generation, from the memoryless properties of the source and the side information channel, from Lemma \ref{lemma:3} with which we bound the term $H(Z^{n}|X^{n},A^{n})$, and from Lemma \ref{lemma:4} with which we bound the term $H(Z^n|A^n,U^n,\mathcal{C}_n)$, and $(d)$ follows from the Markov chain $(Y,Z)-(X,A)-V-U$.

\emph{Secrecy-leakage Rate}: The secrecy-leakage rate averaged over the random codebook $\mathcal{C}_n$ can be bounded as
\begin{align}
	& I(W_k;W_a,W_u,W_v,Z^{n}|\mathcal{C}_n)\nonumber\\
	&\leq H(W_k|\mathcal{C}_n) -H(W_k|W_a,M,W_v,Z^{n},\mathcal{C}_n)\nonumber\\
	&= H(W_k|\mathcal{C}_n) -H(W_a,M,L,Z^{n}|\mathcal{C}_n) \nonumber\\&\qquad+ H(L'|W_a,M,W_v,W_k,Z^n,\mathcal{C}_n) \nonumber\\ &\qquad+ H(W_a,M,W_v,Z^{n}|\mathcal{C}_n)\nonumber\\
	&\overset{(a)}{\leq} H(W_k|\mathcal{C}_n) -H(A^n,U^n,V^n,Z^n|\mathcal{C}_n) + n\epsilon_n \nonumber\\ &\qquad+ H(W_a|\mathcal{C}_n) + H(M|\mathcal{C}_n)+H(W_v|\mathcal{C}_n) \nonumber\\ &\qquad+H(Z^n|A^n,U^n,\mathcal{C}_n)\nonumber\\
	&\overset{(b)}{\leq} H(W_k|\mathcal{C}_n) -H(A^n,U^n,V^n,Z^n|\mathcal{C}_n) + n\epsilon_n \nonumber\\ &\qquad + n(I(X;A)+I(U;X|A)+I(V;X|A,U,Y)\nonumber\\ &\qquad+H(Z|A,U) +\delta_{\epsilon}') \nonumber
	\\&\overset{(c)}{\leq} n\delta_{\epsilon}^{(2)}\label{eq:achsecrecyleakagetheo1}
\end{align}
where $(a)$ follows from the fact that given the codebook, $(A^n,U^n)$ are functions of $(W_a,M)$ and $V^n$ of $(W_a,M,L)$, and from Fano's inequality where given $(W_a,M,W_v,W_k,Z^n)$, the codeword $V^n$ and thus $L'$ can be decoded correctly with high probability since there are less than $2^{nI(V;Z|A,U)}$ remaining $V^n$, $(b)$ follows from the codebook generation and Lemma~\ref{lemma:4}, and $(c)$ follows from the codebook generation,
from the bound on $H(A^n,U^n,V^n,Z^n|\mathcal{C}_n)$ which is shown below,  and from the Markov chain $U-V-(X,A)-(Y,Z)$.
\begin{align*}
	&H(A^n,U^n,V^n,Z^n|\mathcal{C}_n) \\
	&\!\overset{(a)}{=} H(A^n,U^n,V^n,X^n|\mathcal{C}_n) + H(Z^n|X^n,A^n)\\ &\qquad-H(X^n|A^n,U^n,V^n,Z^n,\mathcal{C}_n)\\
	&\!\geq\! H(X^n)\!+\! H(Z^n|X^n,A^n)\!-\!H(X^n|A^n,U^n,V^n,Z^n,\mathcal{C}_n)\\
	&\!\overset{(b)}{\geq} n(H(X) + H(Z|X,A) - H(X|A,U,V,Z) - \delta_{\epsilon}')
\end{align*}
where $(a)$ follows from the Markov chain $Z^n-(X^n,A^n)-(U^n,V^n,\mathcal{C}_n)$ and $(b)$ follows from Lemma~\ref{lemma:3} and from a bound on $H(X^n|A^n,U^n,V^n,Z^n,\mathcal{C}_n)$ which can be derived similarly as in Lemma~\ref{lemma:4}. 

\emph{Secret-key Rate}: The key rate averaged over the random codebook $\mathcal{C}_n$ can be bounded as follows.
\begin{align}
	&H(W_k|\mathcal{C}_n)\geq H(W_k|W_a,M,W_v,L',\mathcal{C}_n) \nonumber\\
	&\overset{(a)}{\geq}  H(A^n,U^n,V^n|\mathcal{C}_n) - H(W_a|\mathcal{C}_n) - H(M|\mathcal{C}_n)\nonumber\\ &\quad-H(W_v|\mathcal{C}_n) -H(L'|\mathcal{C}_n)\nonumber\\
	&\overset{(b)}{\geq} n(I(X;A,U,V) -I(X;A)-I(U;X|A)\nonumber\\ &\quad-I(V;X|A,U,Y)-I(V;Z|A,U)-\delta_{\epsilon}')\nonumber\\
	&= n(I(V;Y|A,U)\!-\!I(V;Z|A,U)\!-\!\delta_{\epsilon}')\geq n(R_k-\delta_{\epsilon}')\label{eq:achsecretkeyrateanalysistheo1}
\end{align}
if $R_k \leq I(V;Y|A,U)-I(V;Z|A,U)$, where $(a)$ follows from the fact that given the codebook $(A^n,U^n,V^n)$ are functions of $(W_a,M,L)$, $(b)$ follows from the codebook generation, from the bound $P_{A^nU^nV^n}(a^n,u^n,v^n) = \sum_{x^n \in \mathcal{T}_{\epsilon}^{(n)}(X|a^n,u^n,v^n)} P_{X^n}(x^n) \leq 2^{-n(I(X;A,U,V)-\delta_{\epsilon})}$, and from the Markov chain $V-(X,A,U)-Y$. 

Using the random coding argument, we have that a tuple $(R_k,R_w,\Delta,C) \in \mathbb{R}_{+}^4$ that satisfies (\ref{eq:keyrate})-(\ref{eq:leakagerate}) for some $P_{A|X}$, $P_{V|XA}$, and $P_{U|V}$ such that $\mathbb{E}[\Gamma(A)] \leq C$ is achievable.

\subsection{Proof of Converse}\label{appsub:converseprooftheo1}
Let $U_{i}\triangleq (W,A^{n \setminus i},Y_{i+1}^{n},Z^{i-1})$ and $V_{i}\triangleq (W,K,A^{n \setminus i},Y_{i+1}^n,Z^{i-1})$, which satisfy the Markov chain $U_i-V_i-(A_i,X_i)-(Y_i,Z_i)$ for all $i=1,2,\ldots,n$. For any achievable tuple $(R_k,R_w,\Delta,C) $, we have the following.

\emph{Storage Rate}: We obtain
\begin{align*}
	&n(R_w+\delta_n) \geq \log|\mathcal{W}^{(n)}| \geq H(W)\\
	&\overset{(a)}{=} H(W)+H(A^{n}|W)= H(A^{n}) + H(W|A^{n})\\
	&\geq [H(A^{n})-H(A^{n}|X^{n},Z^{n})]+ [H(W|A^{n},Y^{n}) \\&\qquad -H(W|A^{n},X^{n},Y^{n},Z^{n}) ] \\
	&\!=\!H(X^{n},Z^{n})\!-\!H(X^{n},Z^{n}|A^{n})\!+\!H(X^{n},Z^{n}|A^{n},Y^{n})\\&\qquad-H(X^{n},Z^{n}|A^{n},Y^{n},W)\\
	&\!=\! H(X^{n})\!+\! H(Z^{n}|X^{n})\!-\!H(Y^{n}|A^{n})\\&\qquad \!+\!H(Y^{n},Z^{n}|X^{n},A^{n})-H(Z^{n}|X^{n},A^{n})\\
	&\qquad -H(X^{n},Z^{n}|A^{n},Y^{n},W,K)\\&\qquad-I(X^n,Z^n;K|A^n,Y^n,W)\\
	&\!\geq\!   H(X^{n}) \!-\!H(Y^{n}|A^{n})\!+\!H(Y^{n},Z^{n}|X^{n},A^{n})\\&\qquad \!-\!H(X^{n},Z^{n}|A^{n},Y^{n},W,K)\!-\!H(K|A^n,Y^n,W)\\
	&\overset{(b)}{\geq}  \sum_{i=1}^{n} H(X_{i}) - H(Y_i|A_{i})+H(Y_i,Z_i|X_{i},A_{i}) \\&\qquad-H(X_{i},Z_i|A^{n},Y^{n},W,K,X^{i-1},Z^{i-1})-n\epsilon_n\\
	&\overset{(c)}{\geq} \sum_{i=1}^{n} H(X_{i}) - H(Y_i|A_{i})+H(Y_i|X_{i},A_{i},Z_i) \\&\qquad+H(Z_i|X_{i},A_{i})-H(X_{i},Z_i|A_{i},Y_i,V_{i})-n\epsilon_n\\
	&\geq \sum_{i=1}^{n} I(X_{i};A_{i})+I(V_{i};X_{i}|A_{i},Y_i)-n\epsilon_n
\end{align*}
where $(a)$ follows from the deterministic action encoder, $(b)$ follows from Fano's inequality, and $(c)$ follows from the definition of $V_i$. 

\emph{Privacy-leakage Rate}: We have
\begin{align*}
	&n(\Delta + \delta_n) \geq I(X^{n};W,Z^{n})
	\\&= I(X^{n};W)+I(X^{n};Z^{n}|W)\\
	& \overset{(a)}{=} I(X^{n};W,A^{n})+I(X^{n};Z^{n}|W,A^{n})\\
	&=\!  H(X^n\!)\!-\!H(X^{n}|W,K,A^{n},Y^{n}\!) \!-\!I(X^n;K|W,A^n,Y^n\!)\\&\qquad-I(X^{n};Y^{n}|W,A^{n})+I(X^{n};Z^{n}|W,A^{n})\\
	& \overset{(b)}{\geq} \sum_{i=1}^{n} H(X_i)-H(X_{i}|W,K,A^{n},Y^{n},X^{i-1})  \\&\qquad -H(Y_i|W,A^{n},Y_{i+1}^{n})+ H(Y_i|X_{i},A_{i}) \\&\qquad +H(Z_i|W,A^{n},Z^{i-1})-H(Z_i|X_{i},A_{i})-n\epsilon_n 
				\end{align*}
				\begin{align*}
	& \overset{(c)}{=} \sum_{i=1}^{n} H(X_i)-H(X_{i}|W,K,A^{n},Y^{n},X^{i-1},Z^{i-1}) \\&\qquad -I(X_{i};Y_i|A_{i})+ H(Y_i|A_{i}) +I(X_{i};Z_i|A_{i})\\ &\qquad - H(Z_i|A_{i}) -H(Y_i|W,A^{n},Y_{i+1}^{n})\\&\qquad+H(Z_i|W,A^{n},Z^{i-1})-n\epsilon_n\\
	& \overset{(d)}{\geq} \sum_{i=1}^{n} H(X_i)-H(X_{i}|V_{i},A_{i},Y_i)-I(X_{i};Y_i|A_{i})\\&\qquad + H(Y_i|A_{i})+I(X_{i};Z_i|A_{i})- H(Z_i|A_{i}) \\ & \qquad \!-\!H(Y_i|W,A^{n},Y_{i+1}^{n})\!+\!H(Z_i|W,A^{n},Z^{i-1})\!-\!n\epsilon_n\\
	&= \sum_{i=1}^{n} \underbrace{I(X_{i};A_{i},V_{i},Y_i) -I(X_{i};Y_i|A_{i})+I(X_{i};Z_i|A_{i})}_{\triangleq P_{i}} \\&\qquad+ I(W,Y_{i+1}^{n},A^{n \setminus i};Y_i|A_{i})\\&\qquad-I(W,Z^{i-1},A^{n \setminus i};Z_i|A_{i})-n\epsilon_n
\end{align*}
where $(a)$ follows from the deterministic action encoder, $(b)$ follows from Fano's inequality and the Markov chain $(W,K,A^{n \setminus i},X^{n \setminus i},Y_{i+1}^{n},Z^{i-1})-(A_{i},X_{i})-(Y_i,Z_i)$, $(c)$ follows from the Markov chain $(X_{i},W,K,A_{i}^n,Y_i^n)-(A^{i-1},X^{i-1})-(Z^{i-1},Y^{i-1})$, and $(d)$ follows from the definition of $V_{i}$ and  the deterministic action encoder. 

By adding the Csisz\'{a}r's sum identity \cite{CsiszarKornerbook}, i.e., 
$\sum_{i=1}^{n} I(Y_i;Z^{i-1}|A^{n},W,Y_{i+1}^{n})-I(Z_i;Y_{i+1}^{n}|A^{n},W,Z^{i-1})=0$, to the right hand side, we get
\begin{align*}
	&n(\Delta  + \delta_n) \geq \sum_{i=1}^{n} P_{i} + I(W,Y_{i+1}^{n},Z^{i-1},A^{n \setminus i};Y_i|A_{i}) \\&\qquad\qquad\qquad - I(W,Y_{i+1}^{n},Z^{i-1},A^{n \setminus i};Z_i|A_{i})-n\epsilon_n\\
	&\overset{(a)}{=}   \sum_{i=1}^{n} I(X_{i};A_{i},V_{i},Y_i) -I(X_{i};Y_i|A_{i})+I(X_{i};Z_i|A_{i}) \\&\qquad + I(U_{i};Y_i|A_{i})- I(U_{i};Z_i|A_{i})-n\epsilon_n\\
	&\overset{(b)}{=} \sum_{i=1}^{n} I(X_{i};A_{i},V_{i},Y_i) - I(X_{i};Y_i|U_{i},A_{i}) \\&\qquad+I(X_{i};Z_i|U_{i},A_{i})-n\epsilon_n,
\end{align*}
where $(a)$ follows from the definitions of $P_i$ and $U_{i}$ and $(b)$ from the Markov chain $U_{i}-(A_{i},X_{i})-(Y_i,Z_i)$. 

\emph{Secret-key Rate}: We obtain 
\begin{align}
	&n(R_k-\delta_n) \leq H(K)\overset{(a)}{\leq} H(K|W,Z^n) +n\delta_n\nonumber\\
	& \overset{(b)}{=}H(K|W,A^n,Z^n)+n\delta_n\nonumber\\
	& \overset{(c)}{\leq}H(K|W,A^n,Z^n) -H(K|W,A^n,Y^n) + 2n\delta_n\nonumber\\
	&=\sum_{i=1}^n I(K;Y_i|W,A^n,Y_{i+1}^n) - I(K;Z_i|W,A^n,Z^{i-1})+ 2n\delta_n\nonumber\\
	&\overset{(d)}{=} \sum_{i=1}^n  I(K;Y_i|W,A^n,Y_{i+1}^n,Z^{i-1}) \nonumber\\&\qquad- I(K;Z_i|W,A^n,Y_{i+1}^n,Z^{i-1})+ 2n\delta_n\nonumber\\
	&\overset{(e)}{=}\! \sum_{i=1}^n  I(V_i;Y_i|A_i,U_i) \!-\! I(V_i;Z_i|A_i,U_i)\!+\! 2n\delta_n\label{eq:theo1keyrateconverseproof}
\end{align}
where $(a)$ follows by (\ref{eq:secrecyleakage_constraint}), $(b)$ follows from the deterministic action encoder, $(c)$ follows from Fano's inequality, $(d)$ follows from Csisz\'{a}r's sum identity, and $(e)$ follows from the definitions of $U_i$ and $V_i$.

\emph{Action Cost}: We have
\begin{align}
	C+\delta_n &\geq \mathbb{E}\big[\Gamma^{(n)}(A^{n})\big] =\frac{1}{n}\sum_{i=1}^{n} \mathbb{E}\big[\Gamma(A_{i})\big]\label{eq:actioncostconverseprooftheo1}.
\end{align}

Finally, we complete the proof by the standard time-sharing argument and letting $\delta_n\rightarrow 0$.

\textit{Cardinality Bounds}: It can be shown by using the support lemma \cite{CsiszarKornerbook} that $\mathcal{U}$ should have $|\mathcal{X}||\mathcal{A}|-1$ elements to preserve $P_{XA}$ and three more to preserve $H(X|U,V,A,Y)$, $I(X;Z|A,U)-I(X;Y|A,U)$, and $I(V;Y|A,U) - I(V;Z|A,U)$. Similarly, the cardinality $|\mathcal{V}|$ can be limited to at most $(|\mathcal{X}||\mathcal{A}|+2)(|\mathcal{X}||\mathcal{A}|+1)$. 

\section*{Proof of Theorem~\ref{theo:originalchosensecretregion}}\label{sec:proofstheo2}
\subsection{Proof of Achievability} \label{appsub:achprooftheo2}
Fix $P_{A|X}$, $P_{V|XA}$, and $P_{U|V}$ such that $\mathbb{E}[\Gamma(A)]\leq C/(1+\epsilon)$. We use the achievability proof of Theorem~\ref{theo:gsvsregion}. Suppose the key $K'=W_{k'}$, generated as in the generated-secret model, has the same cardinality as the embedded key $K=W_{k}$, i.e., $|\mathcal{K}'|=|\mathcal{K}|$. Consider an encoder $f_2^{(n)}$ with inputs $(X^n,K)$ and outputs $W=(K'+K,W')$. Similarly, consider a decoder $g^{(n)}$ with inputs $(Y^n,W)$ and output $\hat{K}=K'+K-\hat{K}'$, where the addition and subtraction operations are modulo-$|\mathcal{K}|$. The decoder of the generated-secret model is used at the decoder to obtain $\hat{K}'$.

\textit{Error Probability}: We have
\begin{align}
	\Pr[K\ne\hat{K}]=\Pr[K'\ne\hat{K}']\label{eq:errorprobabilityachtheo2}
\end{align}
which is small due to the proof of achievability for the generated-secret model.

\textit{Action Cost}: Similar to the generated-secret model, one can show that the expected cost constraint is satisfied with high probability by using the typical average lemma. 

\textit{Privacy-leakage Rate}: We obtain
\begin{align*}
	&I(X^n;W_{a'},W_{u'},W_{v'},W_k+W_{k'},Z^n|\mathcal{C}_n)
	\\&= I(X^n;W_{a'},W_{u'},W_{v'},Z^n|\mathcal{C}_n) 
	\\&\qquad + I(X^n;W_k+W_{k'}|W_{a'},W_{u'},W_{v'},Z^n,\mathcal{C}_n)\\
	&\leq I(X^n;W_{a'},W_{u'},W_{v'},Z^n|\mathcal{C}_n) 
	\\&\qquad + H(W_k+W_{k'}|W_{a'},W_{u'},W_{v'},Z^n,\mathcal{C}_n)
	\\&\qquad - H(W_k+W_{k'}|W_{a'},W_{u'},W_{v'},Z^n,X^n,W_{k'},\mathcal{C}_n)\\
	&\overset{(a)}{\leq} I(X^n;W_{a'},W_{u'},W_{v'},Z^n|\mathcal{C}_n) +\log|\mathcal{K}|-H(W_k)\\
	&\overset{(b)}{\leq}n[I(X;A,V,Y)\!-\!(I(X;Y|A,U)\!-\!I(X;Z|A,U))\!+\!\delta_{\epsilon}']\\&\leq n[\Delta +\delta_{\epsilon}']
\end{align*}
if $\Delta \geq I(X;A,V,Y)-(I(X;Y|A,U)-I(X;Z|A,U))$, where $(a)$ follows because the embedded key $K=W_k$ is independent of other random variables and $(b)$ follows from uniformity of $W_k$ and (\ref{eq:achprivacyleakagegeneratedsecret}).

\textit{Secrecy-leakage Rate}: Observe that
\begin{align*}
	&I(W_k;W_{a'},W_{u'},W_{v'},W_k+W_{k'},Z^n|\mathcal{C}_n)
	\\&= I(W_k;W_{a'},W_{u'},W_{v'},Z^n|\mathcal{C}_n)
	\\&\qquad +I(W_k;W_k+W_{k'}|W_{a'},W_{u'},W_{v'},Z^n,\mathcal{C}_n)\\
	&\overset{(a)}{=}H(W_k+W_{k'}|W_{a'},W_{u'},W_{v'},Z^n,\mathcal{C}_n) \\&\qquad-H(W_{k'}|W_{a'},W_{u'},W_{v'},Z^n,\mathcal{C}_n)\\
	&\leq\log|\mathcal{K}|-H(W_{k'})+I(W_{k'};W_{a'},W_{u'},W_{v'},Z^n|\mathcal{C}_n)
	\\&\overset{(b)}{\leq}n(\delta_n+\delta_{\epsilon}^{(2)})
\end{align*}
where $(a)$ follows because $K=W_k$ is independent of other random variables and $(b)$ follows by (\ref{eq:achsecrecyleakagetheo1}) and (\ref{eq:achsecretkeyrateanalysistheo1}).

\textit{Secret-key Rate}: We have
\begin{align}
	&H(W_k|\mathcal{C}_n)=\log|\mathcal{K}|\geq H(W_{k'}|\mathcal{C}_n)
	\nonumber\\&\overset{(a)}{\geq}n(I(V;Y|A,U)\!-\!I(V;Z|A,U)\!-\!\delta_{\epsilon}')\!\geq\! n(R_k\!-\!\delta_{\epsilon}')\label{eq:achsecretkeyratetheo2}
\end{align}
if $R_k \leq I(V;Y|A,U)-I(V;Z|A,U)$, where $(a)$ follows by (\ref{eq:achsecretkeyrateanalysistheo1}).

\textit{Storage Rate}: The storage rate is the sum of the storage $R_{w'}$ for the generated-secret model and for $K'+K$. We obtain
\begin{align*}
	&R_w \leq R_{w'}+\frac{1}{n}\log|\mathcal{K}|
	\\&\overset{(a)}{=}I(X,A)+I(V;X|A,Y)+6\delta_{\epsilon}+R_k
	\\&\overset{(b)}{\leq}  I(X,A)+I(V;X|A,Y)+6\delta_{\epsilon}
	\\&\qquad+I(V;Y|A,U)-I(V;Z|A,U)
	\\&\overset{(c)}{=} I(X;A,V)-I(U;Y|A)-I(V;Z|A,U)+6\delta_{\epsilon}
\end{align*}
where $(a)$ follows from the storage rate of the generated-secret model, $(b)$ follows by (\ref{eq:achsecretkeyratetheo2}), and $(c)$ follows from the Markov chain $U-V-(X,A)-(Y,Z)$.

Using the random coding argument, we have that a tuple $(R_k, R_w,\Delta,C)\in \mathbb{R}^4_{+}$ that satisfies (\ref{eq:orcssecretrate})-(\ref{eq:orcsleakagerate}) for some $P_{A|X}$, $P_{V|XA}$, and $P_{U|V}$ such that $\mathbb{E}[\Gamma(A)]\!\leq\! C$ is achievable.	

\subsection{Proof of Converse}\label{appsub:converseprooftheo2}
Use the definitions of $U_i$ and $V_i$ given in Appendix~\ref{appsub:converseprooftheo1} so that $U_i-V_i-(A_i,X_i)-(Y_i,Z_i)$ forms a  Markov chain for all $i=1,2,\ldots,n$. The main step is the proof of converse for the storage rate.	

\textit{Secret-key Rate}: Use similar steps as in (\ref{eq:theo1keyrateconverseproof}) to obtain
\begin{align*}
	R_k\!\leq\! \frac{1}{n}\Big[\sum_{i=1}^n  I(V_i;Y_i|A_i,U_i)\! -\! I(V_i;Z_i|A_i,U_i)\!+\! 3n\delta_n\Big].
\end{align*}

\textit{Action Cost}: Similar to Appendix~\ref{appsub:converseprooftheo1}, we obtain (\ref{eq:actioncostconverseprooftheo1}) for the expected cost constraint.

\textit{Privacy-leakage Rate}: We apply similar steps as in Appendix~\ref{appsub:converseprooftheo1} and obtain  
\begin{align*}
	\Delta\geq&\frac{1}{n}\Big[\sum_{i=1}^{n} I(X_{i};V_{i},A_{i},Y_i) - I(X_{i};Y_i|U_{i},A_{i}) \\&\qquad+I(X_{i};Z_i|U_{i},A_{i})-n\epsilon_n-n\delta_n\Big].
\end{align*}

\textit{Storage Rate}: We have
\begin{align*}
	&n(R_w+\delta_n) \geq \log|\mathcal{W}^{(n)}| \geq H(W)\\
	&\overset{(a)}{=} H(W)+H(A^{n}|W)= H(A^{n}) + H(W|A^{n})\\
	&\overset{(b)}{\geq} H(A^{n})-H(A^{n}|X^{n},Z^{n})+H(A^{n}|X^{n},Z^{n})
	\\&\qquad+ H(W|A^{n},Y^{n})  -H(W|A^{n},X^{n},Y^{n},Z^{n})
	\\&\qquad +H(W|A^{n},X^{n})\\
	&=H(X^{n},Z^{n})-H(X^{n},Z^{n}|A^{n})+H(A^{n}|X^{n},Z^{n})
	\\&\qquad+H(X^{n},Z^{n}|A^{n},Y^{n})-H(X^{n},Z^{n}|A^{n},Y^{n},W)
	\\&\qquad+H(W|A^{n},X^{n})\\
	&= H(X^{n}) + H(Z^{n}|X^{n})-H(Y^{n}|A^{n})
	\\&\qquad +H(Y^{n},Z^{n}|X^{n},A^{n})-H(Z^{n}|X^{n},A^{n})
	\\&\qquad+H(A^{n}|X^{n},Z^{n})-H(X^{n},Z^{n}|A^{n},Y^{n},W,K)\\
	&\qquad -I(X^n,Z^n;K|A^n,Y^n,W)+H(W|A^{n},X^{n})\\
	&=   H(X^{n}) + I(Z^n;A^n|X^n)-H(Y^{n}|A^{n})
	\\&\qquad+H(Y^{n},Z^{n}|X^{n},A^{n})+H(A^{n}|X^{n},Z^{n})
	\\&\qquad-H(X^{n},Z^{n}|A^{n},Y^{n},W,K)-H(K|A^n,Y^n,W)
	\\&\qquad + H(K|A^n,Y^n,W,X^n,Z^n)+H(W|A^{n},X^{n})\\
	&\overset{(c)}{=}\!  H(X^{n})\!+\!H(W,A^n,K|X^n)-H(Y^{n}|A^{n})
	\\&\qquad\!+\!H(Y^{n},Z^{n}|X^{n},A^{n})\!-\!H(X^{n},Z^{n}|A^{n},Y^{n},W,K)
	\\&\qquad -H(K|A^n,Y^n,W)\\
	&\overset{(d)}{\geq}\!  H(X^{n})\!+\!H(K)\!-\!H(Y^{n}|A^{n})\!+\!H(Y^{n},Z^{n}|X^{n},A^{n})
	\\&\qquad -H(X^{n},Z^{n}|A^{n},Y^{n},W,K)-H(K|A^n,Y^n,W)\\
	&\geq  H(X^{n})-H(Y^{n}|A^{n})+H(Y^{n},Z^{n}|X^{n},A^{n})
	\\&\qquad -H(X^{n},Z^{n}|A^{n},Y^{n},W,K)
	\\&\qquad+H(K|A^n,Z^n,W)-H(K|A^n,Y^n,W)\\
	&\geq  \sum_{i=1}^{n} H(X_{i}) - H(Y_i|A_{i})+H(Y_i,Z_i|X_{i},A_{i}) \\&\qquad-H(X_{i},Z_i|A^{n},Y^{n},W,K,X^{i-1},Z^{i-1})
	\\&\qquad+ I(K;Y_i|W,A^n,Y_{i+1}^n) - I(K;Z_i|W,A^n,Z^{i-1})
	\\&\overset{(e)}{=}  \sum_{i=1}^{n} H(X_{i}) - H(Y_i|A_{i})+H(Y_i,Z_i|X_{i},A_{i}) \\&\qquad-H(X_{i},Z_i|A^{n},Y^{n},W,K,X^{i-1},Z^{i-1})
	\\&\qquad+ I(K;Y_i|W,A^n,Y_{i+1}^n,Z^{i-1})
	\\&\qquad- I(K;Z_i|W,A^n,Y_{i+1}^n,Z^{i-1})\\
	&\overset{(f)}{\geq} \sum_{i=1}^{n} H(X_{i}) - H(Y_i|A_{i})+H(Y_i|X_{i},A_{i},Z_i) \\&\qquad+H(Z_i|X_{i},A_{i})-H(X_{i},Z_i|A_{i},Y_i,V_{i})
	\\&\qquad+I(V_i;Y_i|A_i,U_i) - I(V_i;Z_i|A_i,U_i)\\
	&\geq \sum_{i=1}^{n} I(X_{i};A_{i})+I(V_{i};X_{i}|Y_i,A_{i})+I(V_i;Y_i|A_i,U_i)
	\\&\qquad - I(V_i;Z_i|A_i,U_i)\\
	&\overset{(g)}{=}\sum_{i=1}^{n} I(X_{i};A_{i},V_{i}) - I(U_i;Y_i|A_i) - I(V_i;Z_i|A_i,U_i)
\end{align*}
where $(a)$ follows from the deterministic action encoder, $(b)$ follows from the Markov chain $W-(A^n,X^n)-(Y^n,Z^n)$, $(c)$ follows from the Markov chain $(K,W)-(A^n,X^n)-(Y^n,Z^n)$, $(d)$ follows because the embedded key $K$ is independent of $X^n$, and $(e)$ follows from Csisz\'{a}r's sum identity. We use the definitions of $U_i$ and $V_i$ in $(f)$, and $(g)$ follows because $U_i-V_i-(A_i,X_i)-(Y_i,Z_i)$ forms a Markov chain for all $i=1,2,\ldots,n$. 

The converse follows by applying the standard time-sharing argument and letting $\delta_n\rightarrow 0$.

\textit{Cardinality Bounds}: We use the support lemma and satisfy the Markov condition $U-V-(A,X)-(Y,Z)$. We therefore preserve $P_{XA}$ by using $|\mathcal{X}||\mathcal{A}|-1$ elements. The bound in (\ref{eq:orcsstoragerate}) for the storage rate can be written as
\begin{align*}
	&I(X;A,V)-I(U;Y|A)-I(V;Z|A,U)
	\\& = I(X;A) \!+\! I(V;X|A,Y)\!+\! I(V;Y|A,U)\!-\! I(V;Z|A,U).
\end{align*}
We thus have to preserve three more expressions, i.e., $I(V;Y|A,U) - I(V;Z|A,U)$, $H(X|U,V,A,Y)$, and $I(X;Z|A,U)-I(X;Y|A,U)$. One can therefore preserve all expressions in Theorem~\ref{theo:originalchosensecretregion} by using an auxiliary random variable $U$ with $|\mathcal{U}|\leq|\mathcal{X}||\mathcal{A}|+2$ and, similarly, $V$ with $|\mathcal{V}|\leq(|\mathcal{X}||\mathcal{A}|+2)(|\mathcal{X}||\mathcal{A}|+1)$. 	


\section*{Proof of Theorem~\ref{theo:hiddensourcegeneratedsecretregion}}\label{sec:proofstheo3}
\subsection{Proof of Achievability} \label{appsub:achprooftheo3}
Consider the codebook generation, encoding, and decoding steps of the generated-secret model with a visible source. Fix $P_{A|\widetilde{X}}$, $P_{V|\widetilde{X}A}$, and $P_{U|V}$ such that $\mathbb{E}[\Gamma(A)]\leq C/(1+\epsilon)$. 

We apply the steps in Appendix~\ref{appsub:achprooftheo1} after replacing every $X$ with $\widetilde{X}$ and every realization $x^n$ with $\tilde{x}^n$. These replacements guarantee that $(\widetilde{X}^n,A^n,U^n,V^n,Y^n)$ are jointly typical with high probability due to standard arguments used in Appendix~\ref{appsub:achprooftheo1} for error analysis. Markov lemma \cite{Elgamalbook} then ensures that $(X^n,\widetilde{X}^n,A^n,U^n,V^n,Y^n)$ are also jointly typical with high probability. 

\textit{Action Cost}: The typical average lemma shows that the expected cost constraint is satisfied with high probability.

\textit{Storage Rate}: After replacing $X$ with $\widetilde{X}$ in Appendix~\ref{appsub:achprooftheo1}, the total storage rate in this case is $R_w=I(\widetilde{X},A)+I(V;\widetilde{X}|A,Y)+6\delta_{\epsilon}$ because $U-V-(A,\widetilde{X})-(A,X)-(Y,Z)$ forms a Markov chain.

\textit{Privacy-leakage Rate}: Consider the leakage about the hidden source averaged over the random codebook $\mathcal{C}_n$. 
\begin{align}
	& I(X^{n};W_a,W_u,W_v,Z^{n}|\mathcal{C}_n)
	\nonumber\\&\leq I(X^{n};W_a,W_u,M^{\prime},W_v,Z^{n}|\mathcal{C}_n)\nonumber\\
	&= I(X^{n};W_a,M,W_v,Z^{n}|\mathcal{C}_n)\nonumber\\
	&= H(X^n|\mathcal{C}_n)-H(X^{n},W_a,M,W_v,Z^{n}|\mathcal{C}_n)
	\nonumber\\&\quad +H(W_a,M,W_v|\mathcal{C}_n)+H(Z^n|W_a,M,W_v,\mathcal{C}_n)\nonumber\\
	&\overset{(a)}{=}\! -H(Z^{n}|X^{n},\mathcal{C}_n)-H(W_a,A^n,M,W_v|X^{n},Z^{n},\mathcal{C}_n)\nonumber\\&\quad +H(W_a,M,W_v|\mathcal{C}_n)+H(Z^n|W_a,M,W_v,\mathcal{C}_n)\nonumber\\
	&=- H(Z^{n}|X^{n},A^{n},\mathcal{C}_n) -I(A^n;Z^n|X^n,\mathcal{C}_n) \nonumber\\&\quad\!-\!H(A^n|X^{n},Z^{n},\!\mathcal{C}_n\!)\!-\!H(W_a,M,W_v|X^{n},Z^{n},A^n,\!\mathcal{C}_n\!)\nonumber\\&\quad+H(W_a,M,W_v|\mathcal{C}_n)+H(Z^n|W_a,M,W_v,\mathcal{C}_n)\nonumber
					\end{align}
					\begin{align}
	& \overset{(b)}{=}- H(Z^{n}|X^{n},A^{n}) -H(A^n|X^n,\mathcal{C}_n) \nonumber\\&\quad-H(W_a,M,W_v,W|X^{n},Z^{n},A^n,\mathcal{C}_n)\nonumber\\&\quad+H(W_a,M,W_v|\mathcal{C}_n)+H(Z^n|W_a,M,W_v,\mathcal{C}_n)\nonumber\\
	&=- H(Z^{n}|X^{n},A^{n}) -H(A^n|X^n,\mathcal{C}_n) 
	\nonumber\\&\quad-H(W_a,M,W_v,W,V^n|X^{n},Z^{n},A^n,\mathcal{C}_n)
	\nonumber\\&\quad +H(V^n|X^{n},Z^{n},A^n,W_a,M,W_v,W,\mathcal{C}_n) \nonumber\\&\quad+H(W_a,M,W_v|\mathcal{C}_n)+H(Z^n|W_a,M,W_v,\mathcal{C}_n)\nonumber\\
	&\overset{(c)}{\leq}-\! H(Z^{n}|X^{n},A^{n})\! -\!H(A^n|X^n,\mathcal{C}_n)
	\nonumber\\&\quad -\!H(V^n|X^{n},Z^{n},A^n,\mathcal{C}_n)\!+\!n\epsilon_n\!+\!H(W_a,M,W_v|\mathcal{C}_n)
	\nonumber\\&\quad+H(Z^n|W_a,M,W_v,\mathcal{C}_n)\nonumber\\
	&\overset{(d)}{\leq}\!-\! H(Z^{n}|X^{n},A^{n})\!-\!H(V^n,A^n|X^{n},\mathcal{C}_n)\!+\!H(W_a|\mathcal{C}_n)
	\nonumber\\&\quad\!+\!H(M|\mathcal{C}_n)\!+\!H(W_v|\mathcal{C}_n)\!+\!H(Z^n|A^n,U^n,\mathcal{C}_n)\!+\! n\epsilon_n\nonumber\\
	&\overset{(e)}{\leq}\!-\! H(Z^{n}|X^{n},A^{n})\!-\!n[H(V,A|X)\!-\!H(V,A|\widetilde{X})\!-\!2\delta_{\epsilon}]
	\nonumber\\&\quad\!+\!H(W_a|\mathcal{C}_n)\!+\!H(M|\mathcal{C}_n)\!+\!H(W_v|\mathcal{C}_n)
	\nonumber\\&\quad\!+\!H(Z^n|A^n,U^n,\mathcal{C}_n)+ n\epsilon_n\nonumber\\
	&\overset{(f)}{\leq}-n[ H(Z|X,A) -H(V,A|X)+H(V,A|\widetilde{X})+7\delta_{\epsilon} \nonumber\\&\quad+I(\widetilde{X};A)+I(\widetilde{X};U|A)\!+\!I(V;\widetilde{X}|A,U)
	\nonumber\\&\quad\!-\!I(V;Y|A,U)\!+\!H(Z|A,U)+\epsilon_n]\nonumber\\
	& \overset{(g)}{=} n[I(\widetilde{X};V,A)\!-\!H(V,A|X)+H(V,A|\widetilde{X})
	\nonumber\\&\quad\!-\!I(V;Y|A,U)+I(X;Z|A,U)+\delta_{\epsilon}^{(3)}]\nonumber \\
	&= n[I(X;V,A)\!-\!I(V;Y|A,U)\!+\!I(X;Z|A,U)\!+\!\delta_{\epsilon}^{(3)}]\nonumber \\
	& \overset{(h)}{=}\! n[I(X;A,V,Y)\!-\!(I(X;Y|A,U)\!-\!I(X;Z|A,U))\!+\!\delta_{\epsilon}^{(3)}]\nonumber\\&\leq n[\Delta +\delta_{\epsilon}^{(3)}]\label{eq:achprivacyleakagetheo3}
\end{align}
if $\Delta \geq I(X;A,V,Y)-(I(X;Y|A,U)-I(X;Z|A,U))$, where $(a)$ follows since given $\mathcal{C}_n$, $W_a$ determines $A^n$,\\
$(b)$ follows since $Z^n-(X^n,A^n)-\mathcal{C}_n$ forms a Markov chain and $(W_a,W_u,W_v)$ determine the helper data $W$,\\
$(c)$ follows from the Markov chain $V^n-(X^n,A^n,W,\mathcal{C}_n)-Y^n$ and Fano's inequality applied as 
\begin{align*}
	&H(V^n|X^{n},Z^{n},A^n,W_a,M,W_v,W,\mathcal{C}_n)
	\\&\!\leq\! H(V^n|X^{n},A^n,W,\mathcal{C}_n)\!\leq\! H(V^n|Y^n,A^n,W,\mathcal{C}_n)\! \leq \!n\epsilon_n,
\end{align*}
$(d)$ follows from the Markov chain $V^n-(X^n,A^n,\mathcal{C}_n)-Z^n$ and from the facts that given the codebook, $W_a$ determines $A^n$ and $(W_a,M)$ determine $U^n$,\\ 
$(e)$ follows from the following inequality
\begin{align*}
	&H(V^n,A^n|X^n,\mathcal{C}_n)
	\\& = H(A^n|X^n,\mathcal{C}_n) + H(V^n|X^n,A^n,\widetilde{X}^n,\mathcal{C}_n)
	\\&\quad+I(V^n;\widetilde{X}^n|X^n,A^n,\mathcal{C}_n)\\	
	&\geq H(\widetilde{X}^n,A^n|X^n,\mathcal{C}_n)-H(\widetilde{X}^n|X^n,A^n,V^n,\mathcal{C}_n)
	\\&\overset{(a)}{\geq}H(\widetilde{X}^n|X^n)-H(\widetilde{X}^n|X^n,A^n,V^n,\mathcal{C}_n)
	\\&\overset{(b)}{\geq} n[H(\widetilde{X}|X)-H(\widetilde{X}|X,A,V)-2\delta_{\epsilon}]
	\\&\overset{(c)}{=}n[H(V,A|X)-H(V,A|\widetilde{X})-2\delta_{\epsilon}]
\end{align*}
where $(a)$ follows since $\widetilde{X}^n-X^n-\mathcal{C}_n$ forms a Markov chain, $(b)$ follows by applying Lemma~\ref{lemma:3} to bound the term $H(\widetilde{X}^n|X^n)$ and Lemma~\ref{lemma:4} to bound the term $H(\widetilde{X}^n|X^n,A^n,V^n,\mathcal{C}_n)$, and $(c)$ follows due to the Markov chain $(V,A)-\widetilde{X}-X$, 
\\ 
$(f)$ follows from the codebook generation, from the memoryless property of the source and side information channels, from Lemma~\ref{lemma:3} applied to $H(Z^n|X^n,A^n)$, and from Lemma~\ref{lemma:4} applied to $H(Z^n|A^n,U^n,\mathcal{C}_n)$,\\
$(g)$ follows from the Markov chains $U-(V,A)-\widetilde{X}$ and $U-(A,X)-Z$,\\
$(h)$ follows from the Markov chain $U-V-(A,X)-Y$.

\textit{Secrecy-leakage Rate}: The secrecy-leakage rate analysis follows by replacing every $X^n$ in Appendix~\ref{appsub:achprooftheo1} with $\widetilde{X}^n$ when bounding the term $H(A^n,U^n,V^n,Z^n|\mathcal{C}_n)$ since, this time, $(U^n,V^n,\mathcal{C}_n)-(A^n,\widetilde{X}^n)-Z^n$ and $U-V-(A,\widetilde{X})-(Y,Z)$ form Markov chains. Use
\begin{align*}
	&H(Z^n|\widetilde{X}^n,A^n,\mathcal{C}_n)
	\\&= H(Z^n|\widetilde{X}^n,A^n,X^n,\mathcal{C}_n)+I(Z^n;X^n|\widetilde{X}^n,A^n,\mathcal{C}_n)
	\\&\!\overset{(a)}{=}\!H(Z^n|A^n,X^n)\!+\!H(X^n|\widetilde{X}^n)\!-\!H(X^n|\widetilde{X}^n,\!A^n,Z^n,\!\mathcal{C}_n)
	\\&\!\overset{(b)}{\geq}\!n(H(Z|A,\!X)\!+\!H(X|\widetilde{X})\!-\!2\delta_{\epsilon})\!-\!H(X^n|\widetilde{X}^n,\!A^n,\!Z^n,\!\mathcal{C}_n)
	\\&\overset{(c)}{\geq}\!n(H(Z|A,X)\!+\!H(X|\widetilde{X},A)\!-\!H(X|\widetilde{X},A,Z)\!-\!3\delta_{\epsilon})
	\\&\overset{(d)}{=}n(H(Z|\widetilde{X},A)-3\delta_{\epsilon})
\end{align*}
where $(a)$ follows because $Z^n-(A^n,X^n)-(\widetilde{X}^n,\mathcal{C}_n)$ and $X^n-\widetilde{X}^n-(A^n,\mathcal{C}_n)$ form Markov chains, $(b)$ follows by applying Lemma~\ref{lemma:3} to bound the terms $H(Z^n|A^n,X^n)$ and $H(X^n|\widetilde{X}^n)$ because $Z^n$ is i.i.d. $\sim P_{Z|XA}$ and $X^n$ is i.i.d. $\sim P_{X|\widetilde{X}}$, $(c)$ follows from the Markov chain $X-\widetilde{X}-A$ and by applying Lemma~\ref{lemma:4} to bound the term $H(X^n|\widetilde{X}^n,A^n,Z^n,\mathcal{C}_n)$, and $(d)$ follows from the Markov chain $Z-(A,X)-\widetilde{X}$. We thus obtain
\begin{align}
	I(W_k;W_a,W_u,W_v,Z^n|\mathcal{C}_n)\leq n\delta_{\epsilon}^{(4)}\label{eq:achsecrecyleakagetheo3}.
\end{align}

\textit{Secret-key Rate}: Using the codebook generation in Appendix~\ref{appsub:achprooftheo1} and the fact that now $V-(A,\widetilde{X},U)-Y$ forms a Markov chain, it is straightforward to show that 
\begin{align}
	&H(W_k|\mathcal{C}_n)\geq n[I(V;Y|A,U)-I(Y;Z|A,U)-\delta_{\epsilon}^{(3)}]\nonumber\\
	&\geq n(R_k-\delta_{\epsilon}^{(3)})\label{eq:achsecretkeyratetheo3}
\end{align}
if $R_k\leq I(V;Y|A,U)-I(V;Z|A,U)$.

Using the random coding argument, we have that a tuple $(R_k, R_w,\Delta,C)\in \mathbb{R}^4_{+}$ that satisfies (\ref{eq:hiddensecretrate})-(\ref{eq:hiddenleakagerate}) for some $P_{A|\widetilde{X}}$, $P_{V|\widetilde{X}A}$, and $P_{U|V}$ such that $\mathbb{E}[\Gamma(A)]\!\leq\! C$ is achievable.

\subsection{Proof of Converse}\label{appsub:converseprooftheo3}
Use the definitions of $U_i$ and $V_i$ given in Appendix~\ref{appsub:converseprooftheo1} so that $U_i-V_i-(A_i,\widetilde{X}_i)-(A_i,X_i)-(Y_i,Z_i)$ forms a  Markov chain for all $i=1,2,\ldots,n$.

\textit{Storage Rate}: Replace every $X^n$ with $\widetilde{X}^n$ and every $X_i$ with $\widetilde{X}_i$ for all $i=1,2,\ldots,n$ in Appendix~\ref{appsub:converseprooftheo1} and apply similar steps to obtain 
\begin{align*}
	R_w\!\geq\! \frac{1}{n}\Big[\sum_{i=1}^{n} I(\widetilde{X}_{i};A_{i})\!+\!I(V_{i};\widetilde{X}_{i}|A_{i},Y_i)\!-\!n\epsilon_n\!-\!n\delta_n\Big].
\end{align*}

\textit{Privacy-leakage Rate}: We apply similar steps as in Appendix~\ref{appsub:converseprooftheo1}. It is also straightforward to show that $(W,K,A^{n \setminus i},X^{n \setminus i},Y_{i+1}^{n},Z^{i-1})-(A_{i},X_{i})-(Y_i,Z_i)$, $(X_{i},W,K,A_{i}^n,Y_i^n)-(A^{i-1},X^{i-1})-(Z^{i-1},Y^{i-1})$, and $U_{i}-(A_{i},X_{i})-(Y_i,Z_i)$ form Markov chains for all $i=1,2,\ldots,n$ also for a hidden source. We thus obtain
\begin{align*}
	\Delta\geq&\frac{1}{n}\Big[\sum_{i=1}^{n} I(X_{i};A_{i},V_{i},Y_i) - I(X_{i};Y_i|A_{i},U_{i}) \\&\qquad+I(X_{i};Z_i|A_{i},U_{i})-n\epsilon_n-n\delta_n\Big].
\end{align*}

\textit{Secret-key Rate}: The converse is similar to the converse for a visible source with the generated-secret model. By applying similar steps as in Appendix~\ref{appsub:converseprooftheo1}, we obtain
\begin{align*}
	R_k\!\leq\! \frac{1}{n}\Big[\sum_{i=1}^n  I(V_i;Y_i|A_i,U_i) \!-\! I(V_i;Z_i|A_i,U_i)\!+\! 3n\delta_n\Big].
\end{align*}

\textit{Action Cost}: We obtain (\ref{eq:actioncostconverseprooftheo1}) for the expected cost constraint. 	

The converse follows by applying the standard time-sharing argument and letting $\delta_n\rightarrow 0$.

\textit{Cardinality Bounds}: We use the support lemma and satisfy the Markov condition $U-V-(A,\widetilde{X})-(A,X)-(Y,Z)$, so we preserve $P_{\widetilde{X}A}$ by using $|\mathcal{\widetilde{X}}||\mathcal{A}|-1$ real-valued continuous functions. We have to preserve four more expressions, i.e., $I(V;Y|A,U) - I(V;Z|A,U)$, $H(\widetilde{X}|U,V,A,Y)$, $H(X|U,V,A,Y)$, and $I(X;Z|A,U)-I(X;Y|A,U)$. One can therefore preserve all expressions in Theorem~\ref{theo:hiddensourcegeneratedsecretregion} by using an auxiliary random variable $U$ with $|\mathcal{U}|\leq|\mathcal{X}||\mathcal{A}|+3$ and, similarly, $V$ with $|\mathcal{V}|\leq(|\mathcal{X}||\mathcal{A}|+3)(|\mathcal{X}||\mathcal{A}|+2)$.


\section*{Proof of Theorem~\ref{theo:hiddensourcechosensecretregion}}\label{sec:proofstheo4}
\subsection{Proof of Achievability}\label{appsub:achprooftheo4}
Fix $P_{A|\widetilde{X}}$, $P_{V|\widetilde{X}A}$, and $P_{U|V}$ such that $\mathbb{E}[\Gamma(A)]\leq C/(1+\epsilon)$. We use the achievability proof of Theorem~\ref{theo:hiddensourcegeneratedsecretregion}. Suppose the key $K'=W_{k'}$ generated as in the generated-secret model for a hidden source has the same cardinality as the embedded key $K=W_{k}$, i.e., $|\mathcal{K}'|=|\mathcal{K}|$. Consider an encoder $f_4^{(n)}$ with inputs $(\widetilde{X}^n,K)$ and outputs $W=(K'+K,W')$. Similarly, consider a decoder $g^{(n)}$ with inputs $(Y^n,W)$ and output $\hat{K}=K'+K-\hat{K}'$, where the addition and subtraction operations are modulo-$|\mathcal{K}|$. Note that the decoder of the generated-secret model for a hidden source is used at the decoder to obtain $\hat{K}'$.

\textit{Error Probability}: We obtain (\ref{eq:errorprobabilityachtheo2}), which is small due to the proof of achievability for Theorem~\ref{theo:hiddensourcegeneratedsecretregion}.

\textit{Action Cost}: Similar to Appendix~\ref{appsub:achprooftheo3}, one can show that the expected cost constraint is satisfied with high probability by using the typical average lemma. 

\textit{Privacy-leakage Rate}: We have
\begin{align*}
	&I(X^n;W_{a'},W_{u'},W_{v'},W_k+W_{k'},Z^n|\mathcal{C}_n)
	\\&\leq I(X^n;W_{a'},W_{u'},W_{v'},Z^n|\mathcal{C}_n)+\log|\mathcal{K}|
	\\&\qquad -H(W_k+W_{k'}|W_{a'},W_{u'},W_{v'},Z^n,X^n,W_{k'},\mathcal{C}_n)\\
	&\overset{(a)}{\leq}n[I(X;A,V,Y)\!-\!(I(X;Y|A,U)\!-\!I(X;Z|A,U))\!+\!\delta_{\epsilon}^{(3)}]\\
	&\leq n[\Delta +\delta_{\epsilon}^{(3)}]
\end{align*}
if $\Delta \geq I(X;A,V,Y)-(I(X;Y|A,U)-I(X;Z|A,U))$, where $(a)$ follows because $K=W_k$ is independent of other random variables, and from uniformity of $W_k$ and (\ref{eq:achprivacyleakagetheo3}).

\textit{Secrecy-leakage Rate}: We obtain
\begin{align*}
	&I(W_k;W_{a'},W_{u'},W_{v'},W_k+W_{k'},Z^n|\mathcal{C}_n)
	\\&= I(W_k;W_{a'},W_{u'},W_{v'},Z^n|\mathcal{C}_n)
	\\&\qquad +I(W_k;W_k+W_{k'}|W_{a'},W_{u'},W_{v'},Z^n,\mathcal{C}_n)
	\\&\overset{(a)}{=}H(W_k+W_{k'}|W_{a'},W_{u'},W_{v'},Z^n,\mathcal{C}_n) \\&\qquad-H(W_{k'}|W_{a'},W_{u'},W_{v'},Z^n,\mathcal{C}_n)
	\\&\leq\log|\mathcal{K}|-H(W_{k'})+I(W_{k'};W_{a'},W_{u'},W_{v'},Z^n|\mathcal{C}_n)\\
	&\overset{(b)}{\leq}n(\delta_n+\delta_{\epsilon}^{(4)})
\end{align*}
where $(a)$ follows because $K=W_k$ is independent of other random variables and $(b)$ follows by (\ref{eq:achsecrecyleakagetheo3}) and (\ref{eq:achsecretkeyratetheo3}).

\textit{Secret-key Rate}: Observe that
\begin{align}
	&H(W_k|\mathcal{C}_n)=\log|\mathcal{K}|\geq H(W_{k'}|\mathcal{C}_n)
	\nonumber\\&\overset{(a)}{\geq}n(I(Y;V|A,U)-I(Z;V|A,U)-\delta_{\epsilon}^{(3)})\nonumber\\&\geq n(R_k-\delta_{\epsilon}^{(3)})\label{eq:achsecretkeyrattheo4}
\end{align}
if $R_k \leq I(V;Y|A,U)-I(V;Z|A,U)$, where $(a)$ follows by (\ref{eq:achsecretkeyratetheo3}).

\textit{Storage Rate}: The storage rate is the sum of the storage $R_{w'}$ for a hidden source with the generated-secret model and for $K'+K$. We obtain
\begin{align*}
	&R_w \leq R_{w'}+\frac{1}{n}\log|\mathcal{K}|
	\\&\overset{(a)}{=}I(\widetilde{X},A)+I(V;\widetilde{X}|A,Y)+6\delta_{\epsilon}+R_k\\
	&\overset{(b)}{\leq}  I(\widetilde{X},A)+I(V;\widetilde{X}|A,Y)+6\delta_{\epsilon}
	\\&\qquad+I(V;Y|A,U)-I(V;Z|A,U)\\
	&\overset{(c)}{=} I(\widetilde{X};A,V)-I(U;Y|A)-I(V;Z|A,U)+6\delta_{\epsilon}
\end{align*}
where $(a)$ follows from the storage rate for a hidden source with the generated-secret model, $(b)$ follows by (\ref{eq:achsecretkeyrattheo4}), and $(c)$ follows from the Markov chain $U-V-(A,\widetilde{X})-(Y,Z)$.

Using the random coding argument, we have that a tuple $(R_k, R_w,\Delta,C)\in \mathbb{R}^4_{+}$ that satisfies (\ref{eq:hiddenchosensecretrate})-(\ref{eq:hiddenchosenleakagerate}) for some $P_{A|\widetilde{X}}$, $P_{V|\widetilde{X}A}$, and $P_{U|V}$ such that $\mathbb{E}[\Gamma(A)]\!\leq\! C$ is achievable.

\subsection{Proof of Converse}\label{appsub:converseprooftheo4}
Use the definitions of $U_i$ and $V_i$ given in Appendix~\ref{appsub:converseprooftheo1} so that $U_i-V_i-(A_i,\widetilde{X}_i)-(A_i,X_i)-(Y_i,Z_i)$ forms a  Markov chain for all $i=1,2,\ldots,n$.  

\textit{Secret-key Rate}: The converse for the secret-key rate is similar to the converse for a hidden source with the generated-secret model. We obtain
\begin{align*}
	R_k\!\leq\! \frac{1}{n}\Big[\sum_{i=1}^n  I(V_i;Y_i|A_i,U_i) \!-\! I(V_i;Z_i|A_i,U_i)\!+\! 3n\delta_n\Big].
\end{align*}

\textit{Action Cost}: Similar to Appendix~\ref{appsub:converseprooftheo3}, we obtain (\ref{eq:actioncostconverseprooftheo1}) for the expected cost constraint.

\textit{Privacy-leakage Rate}: We apply similar steps to Appendix~\ref{appsub:converseprooftheo3}. It is straightforward to show that $(W,K,A^{n \setminus i},X^{n \setminus i},Y_{i+1}^{n},Z^{i-1})-(A_{i},X_{i})-(Y_i,Z_i)$, $(X_{i},W,K,A_{i}^n,Y_i^n)-(A^{i-1},X^{i-1})-(Z^{i-1},Y^{i-1})$, and $U_{i}-(A_{i},X_{i})-(Y_i,Z_i)$ form Markov chains for all $i=1,2,\ldots,n$ also for a hidden source and an embedded secret key $K$. We thus obtain
\begin{align*}
	\Delta\geq&\frac{1}{n}\Big[\sum_{i=1}^{n} I(X_{i};A_{i},V_{i},Y_i) - I(X_{i};Y_i|A_{i},U_{i}) \\&\qquad+I(X_{i};Z_i|A_{i},U_{i})-n\epsilon_n-n\delta_n\Big].
\end{align*}

\textit{Storage Rate}: This time, we apply similar steps as in Appendix~\ref{appsub:converseprooftheo2}. Replace every sequence $X^n$ with $\widetilde{X}^n$ and every $X_i$ with $\widetilde{X}_i$ for all $i=1,2,\ldots,n$. Using similar steps as in Appendix~\ref{appsub:converseprooftheo2}, and the facts that  $U_i-V_i-(A_i,\widetilde{X}_i)-(Y_i,Z_i)$  for all $i=1,2,\ldots,n$ and $(K,W)-(A^n,\widetilde{X}^n)-(Y^n,Z^n)$ form Markov chains, we obtain 
\begin{align*}
	R_w\geq& \frac{1}{n}\Big[\sum_{i=1}^{n} I(\widetilde{X}_{i};A_{i},V_{i}) - I(U_i;Y_i|A_i)
	\\&\qquad - I(V_i;Z_i|A_i,U_i)-n\delta_n\Big].
\end{align*}

The converse follows by applying the standard time-sharing argument and letting $\delta_n\rightarrow 0$.

\textit{Cardinality Bounds}: We use the support lemma. One has to satisfy the Markov condition $U-V-(A,\widetilde{X})-(A,X)-(Y,Z)$. We therefore preserve $P_{\widetilde{X}A}$ by using $|\mathcal{\widetilde{X}}||\mathcal{A}|-1$ real-valued continuous functions. The bound in (\ref{eq:storagehiddensourcechosensecret}) can be written as
\begin{align*}
	&I(\widetilde{X};A,V)-I(U;Y|A)-I(V;Z|A,U)
	\\& = I(\widetilde{X};A) \!+\! I(\widetilde{X};V|A,Y)\!+\! I(V;Y|A,U)\!-\! I(V;Z|A,U).
\end{align*}
We therefore have to preserve four more expressions, i.e., $I(V;Y|A,U) - I(V;Z|A,U)$, $H(\widetilde{X}|U,V,A,Y)$, $H(X|U,V,A,Y)$, and $I(X;Z|A,U)-I(X;Y|A,U)$. One can therefore preserve all expressions in Theorem~\ref{theo:hiddensourcechosensecretregion} by using an auxiliary random variable $U$ with $|\mathcal{U}|\leq|\mathcal{X}||\mathcal{A}|+3$ and, similarly, $V$ with $|\mathcal{V}|\leq(|\mathcal{X}||\mathcal{A}|+3)(|\mathcal{X}||\mathcal{A}|+2)$.

\ifCLASSOPTIONcaptionsoff
  \newpage
\fi

\bibliographystyle{IEEEtran}
\bibliography{IEEEabrv,references}
%



\begin{IEEEbiography}[{\includegraphics[width=1.05in,height=1.22in]{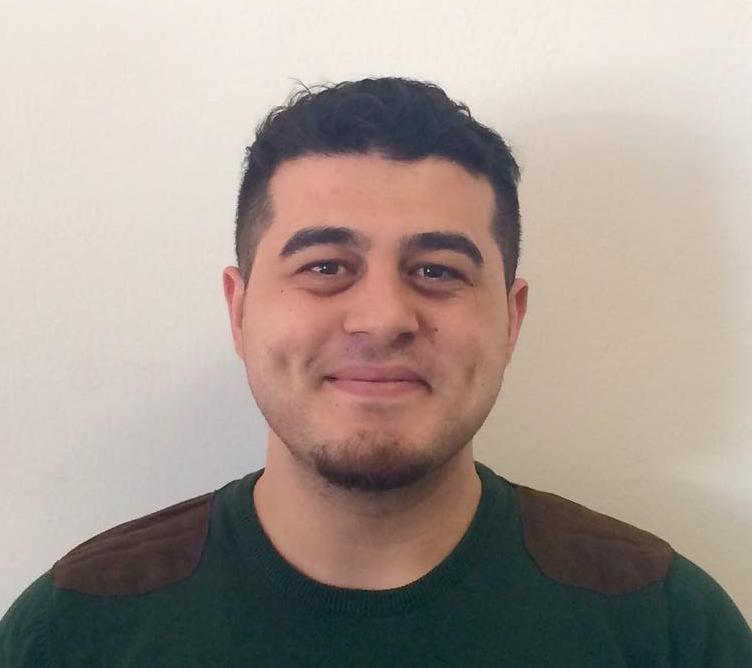}}]{Onur G\"unl\"u}
(S'10) received the B.Sc. degree in electrical and electronics engineering from Bilkent University, Ankara, in 2011, and the M.Sc. degree in communications engineering from the Technical University of Munich (TUM), Munich, in 2013, where he is currently pursuing the Dr.-Ing. degree. He is a Research and Teaching Assistant with TUM. In 2018, he was visiting the Information and Communication Theory Lab, TU Eindhoven, The Netherlands. His research interests include information theoretic privacy and security, code design for secret key generation from the source model, statistical signal processing for biometric secrecy systems and physical unclonable functions (PUFs).
\end{IEEEbiography}

\begin{IEEEbiography}[{\includegraphics[width=1in,height=1.25in]{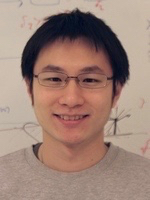}}]{Kittipong Kittichokechai}
	(S'10--M'15) received the B.Eng. degree in electrical engineering from Chulalongkorn University, Thailand, in 2007, and the M.Sc. and Ph.D. degrees in electrical engineering from the KTH Royal Institute of Technology, Sweden, in 2009 and 2014, respectively. In 2012, he was a Visiting Scholar at the Information Systems Laboratory (ISL), Stanford University, USA. From 2014 to 2016, he was a Post-Doctoral Researcher with Technische Universit\"{a}t Berlin, Germany. Since 2016, he has been a Researcher with Ericsson Research, Stockholm, Sweden, where he has been contributing to the development of new communication technologies of 5G. His research interests include network information theory, information theoretic security and privacy, distributed detection, and their applications in wireless communications. 
	
	K. Kittichokechai was a recipient of the Ananda Mahidol Foundation Scholarship under the Royal Patronage of His Majesty the King of Thailand. 
\end{IEEEbiography}

\begin{IEEEbiography}
	[{\includegraphics[width=1in,height=1.25in,clip,keepaspectratio]{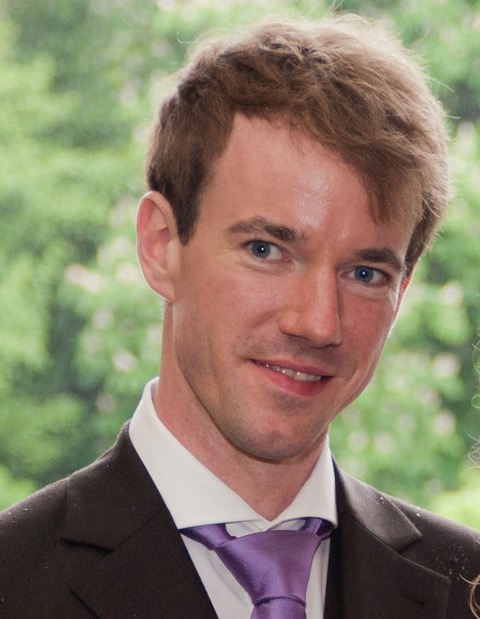}}]{Rafael F. Schaefer}
	(S'08--M'12--SM'17) received the Dipl.-Ing. degree in electrical engineering and computer science from Technische Universit\"at Berlin, Germany, in 2007, and the Dr.-Ing. degree in electrical engineering from Technische Universit\"at M\"unchen, Germany, in 2012. From 2007 to 2010, he was a Research and Teaching Assistant with Technische Universit\"at Berlin and from 2010 to 2013, with Technische Universit\"at M\"unchen. From 2013 to 2015, he was a Post-Doctoral Research Fellow with Princeton University. Since 2015, he has been an Assistant Professor with Technische Universit\"at Berlin. Among his publications is the recent book \emph{Information Theoretic Security and Privacy of Information Systems} (Cambridge University Press, 2017). He is an Associate Member of the IEEE Information Forensics and Security Technical Committee. He was a recipient of the VDE Johann-Philipp-Reis Prize in 2013. He received the best paper award of the German Information Technology Society (ITG-Preis) in 2016. He was one of the exemplary reviewers of the \textsc{IEEE Communication Letters} in 2013. He is currently an Associate Editor of the \textsc{IEEE Transactions on Communications}.
\end{IEEEbiography}

\vspace*{6.8cm}
\begin{IEEEbiography}[{\includegraphics[width=1in,height=1.25in,clip,keepaspectratio]{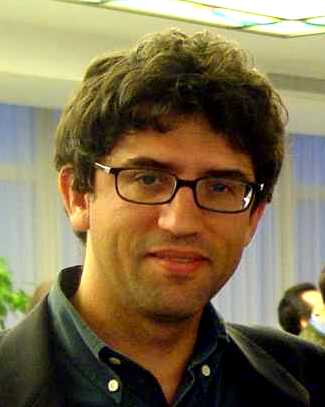}}]{Giuseppe Caire}
(S'92--M'94--SM'03--F'05) was born in Torino, Italy, in 1965. He received the B.Sc. in electrical engineering from the Politecnico di Torino, Italy, in 1990, the M.Sc. in electrical engineering from Princeton University in 1992, and the Ph.D. from the Politecnico di Torino in 1994. He was a Post-Doctoral Research Fellow with the European Space Agency, ESTEC, Noordwijk, The Netherlands, from 1994 to 1995, an Assistant Professor in telecommunications with the Politecnico di Torino, an Associate Professor with the University of Parma, Italy, a Professor with the Department of Mobile Communications, Eurecom Institute, Sophia-Antipolis, France, a Professor of electrical engineering with the Viterbi School of Engineering, University of Southern California, Los Angeles, CA, USA, and is currently an Alexander von Humboldt Professor with the Electrical Engineering and Computer Science Department, Technische Universit\"at Berlin, Germany. 

His main research interests include communications theory, information theory, channel and source coding with particular focus on wireless communications. He served as Associate Editor for the \textsc{IEEE Transactions on Communications} from 1998 to 2001 and as Associate Editor for the \textsc{IEEE Transactions on Information Theory} from 2001 to 2003. He received the Jack Neubauer Best System Paper Award from the IEEE Vehicular Technology Society in 2003, the IEEE Communications Society \& Information Theory Society Joint Paper Award in 2004 and in 2011, the Okawa Research Award in 2006, the Alexander von Humboldt Professorship in 2014, and the Vodafone Innovation Prize in 2015. Giuseppe Caire is a Fellow of IEEE since 2005. He has served on the Board of Governors of the IEEE Information Theory Society from 2004 to 2007, and as an officer from 2008 to 2013. He was President of the IEEE Information Theory Society in 2011.
\end{IEEEbiography}





\end{document}